\documentclass[12pt,onecolumn]{IEEEtran}
\usepackage[margin=1in]{geometry} % guarantees ~6.5" line length
\usepackage{setspace}
\usepackage[T1]{fontenc}

\usepackage{newtxtext,newtxmath}
\usepackage{amsmath,mathtools,bm,dsfont}
\usepackage{graphicx}
\usepackage{booktabs}
\usepackage[caption=false,font=footnotesize]{subfig} % use subfig, not subcaption
\usepackage[linesnumbered,ruled,vlined]{algorithm2e}
\usepackage{float}
\usepackage{cite}
\usepackage{xcolor}

\usepackage{soul}   % defines \hl
\sethlcolor{yellow} % optional
\usepackage[hidelinks]{hyperref}
\usepackage{amsthm}
\usepackage[switch]{lineno}
\newtheorem{theorem}{Theorem}[section]

\theoremstyle{definition}

\theoremstyle{remark}

 \title {
Two-Phase Cell Switching in 6G vHetNets: Sleeping-Cell Load Estimation and Renewable-Aware Switching Toward NES}

\begin{document}

 \author{\rm{Maryam Salamatmoghadasi}, \rm{Metin Ozturk}, \rm{Halim Yanikomeroglu}
 \thanks{This research has been sponsored in part by the NSERC Create program entitled TrustCAV and in part by The Scientific and Technological Research Council of Türkiye (TUBITAK).
 
Maryam Salamatmoghadasi and Halim Yanikomeroglu are with Non-Terrestrial Networks Lab, Department of Systems and Computer Engineering, Carleton University, Ottawa, ON K1S5B6, Canada. Metin Ozturk is with Electrical and Electronics Engineering, Ankara Yildirim Beyazit University, Ankara, 06010, Turkiye. emails: \texttt{maryamsalamatmoghad@cmail.carleton.ca, metin.ozturk@aybu.edu.tr, halim@sce.carleton.ca}.}}

\maketitle
%\linenumbers
\begin{abstract}
This paper proposes a two-phase framework to improve the sustainability in vertical heterogeneous networks (vHetNets) that integrate various types of base stations~(BSs), including terrestrial macro BSs~(MBSs), small BSs~(SBSs), and a high-altitude platform station–super MBS (HAPS-SMBS). 
In Phase I, we address the critical and often-overlooked challenge of estimating the traffic load of sleeping SBSs, a prerequisite for practical cell switching, by introducing three methods with varying data dependencies: (i) a distance-based estimator (no historical data), (ii) a multi-level clustering (MLC) estimator (limited historical data), and (iii) a long short-term memory~(LSTM)-based temporal predictor (full historical data). 
In Phase II, we incorporate the most accurate estimation results from Phase I into a renewable energy–aware cell switching strategy, explicitly modeling solar-powered SBSs in three operational scenarios that reflect realistic hybrid grid–renewable deployments. 
This flexible design allows the framework to adapt switching strategies based on renewable availability and storage conditions, making it more practical and robust for real-world networks. Using a real call detail record dataset from Milan, simulation results show that the LSTM method achieves a mean absolute percentage error (MAPE) below 1\% in Phase I, while in Phase II, the threshold-based solar integration scenario achieves up to 23\% network energy saving (NES) relative to conventional cell switching. Overall, the proposed framework bridges the gap between theoretical cell switching models and practical, sustainable 6G radio access network~(RAN) operation, enabling significant energy saving without compromising quality of service.

\end{abstract}
\section{Introduction}
%\subsection{Background}
The sixth-generation (6G) cellular networks are envisioned to support hyperconnectivity with up to 10 million devices per square kilometer, which is rapidly approaching reality~\cite{bs_power}. 
While this evolution promises ubiquitous connectivity and advanced services, it also raises serious sustainability concerns, particularly in terms of escalating energy consumption across radio access networks (RANs). Base stations (BSs), especially dense deployments of small BSs (SBSs), are responsible for approximately 60–80\% of a network's total power use~\cite{bs_power}. If not managed effectively, the energy demand of 6G networks could exceed that of fifth-generation (5G) by more than an order of magnitude~\cite{2222222}, posing a threat to both environmental and economic sustainability.

Addressing this challenge is pivotal for meeting sustainability targets. Reports from the International Telecommunication Union (ITU) and the European Commission note that future information and communication technology (ICT) systems must support global connectivity while reducing carbon intensity and absolute energy use~\cite{itu_vision_june_23}. Achieving net-zero trajectories for mobile operators by 2040 requires innovations in energy-efficient RAN design and operation~\cite{Belkhir2018}. 
In parallel, Network Energy Saving (NES) has been formalized in the 3rd Generation Partnership
Project (3GPP) framework as a cross-RAN objective, standardizing energy-control functions such as BS sleep states, power scaling, cell/beam muting, and on-demand broadcast signaling~\cite{3GPP}.

Among energy-saving approaches, dynamically switching off idle or underutilized BSs\textemdash\allowbreak commonly referred to as \textit{cell switching} or \textit{BS sleep mode}—has emerged as a key strategy~\cite{6512843}. By deactivating SBSs during low traffic periods and redistributing the load across active BSs, networks can reduce unnecessary energy usage without significantly compromising the user quality of service (QoS). However, a fundamental bottleneck is the lack of accurate traffic load information for sleeping SBSs. Since deactivated BSs are disconnected from the network scheduler, their potential traffic demand, a critical information for the optimization process, is inherently unknown. Most studies in the literature assume perfect knowledge of traffic loads—even for SBSs in sleep mode—which is unrealistic in operational environments~\cite{ELAA2022JR, EOMK2017JR, Metin_VFA_CellSwitch, 8735834, ACM}. This disconnect limits the real-world applicability of even the most advanced switching algorithms and optimization frameworks.
Bridging this gap requires accurate \textit{estimation} of traffic loads at sleeping SBSs for the next time slot, a problem that is often overlooked but critically important. Without reliable estimates, decisions regarding which SBSs to activate/deactivate may lead to performance degradation, including congestion, suboptimal energy savings, or degraded QoS. Thus, the load estimation problem is a significant obstacle between theoretical frameworks and the practical deployment of cell switching policies in real-world networks.

In parallel, integrating renewable energy—particularly solar—into RAN infrastructure offers an additional path to sustainability. Powering SBSs with on-site photovoltaics and storage can reduce grid dependence and emissions. However, renewable energy sources are intermittent and unpredictable, motivating energy-aware optimization under uncertain supply~\cite{ALOTAIBI2025104213}. Deactivating a solar-powered SBS with sufficient stored energy may result in increased grid energy usage, counteracting energy-saving goals. Therefore, effective energy management must jointly consider traffic demand and energy availability.

Together, these two challenges—estimating the loads of sleeping SBSs and making renewable-aware switching decisions—are central to realizing energy-efficient and sustainable 6G mobile networks. Therefore, we tackle both challenges with a two-phase approach that estimates sleeping-cell loads and enables renewable-aware cell switching.

\subsection{Related Work}

Extensive research has targeted RAN energy reduction via cell switching in heterogeneous deployments~\cite{ELAA2022JR, EOMK2017JR, Metin_VFA_CellSwitch, 8735834, ACM, 8024181, 10061602}. 
In~\cite{ELAA2022JR}, a decentralized control mechanism was proposed to dynamically adjust BS sleep depth based on activity levels, enabling scalable cell switching implementations in dense deployments. 
Similarly, the authors in~\cite{EOMK2017JR} explored energy savings through a control-data separation architecture (CDSA), optimizing user association and BS on/off decisions using genetic algorithms.
More recently, reinforcement learning~(RL) was applied in~\cite{Metin_VFA_CellSwitch}, where a value function approximation (VFA)-based RL method achieved significant energy savings in ultra-dense heterogeneous networks (HetNets).
In~\cite{9528008}, the authors investigated a BS switching-off strategy based on traffic prediction, considering environmental variability across heterogeneous BS types. 
Other works include traffic-aware partial activation~\cite{8735834}, strategy families that reduce operational energy while preserving coverage~\cite{ACM}, and power-adjustable cell switching where SBSs are selectively switched off while macro BSs (MBSs) maintain control plane coverage~\cite{8024181}.
A beamforming-aware cell switching strategy for multi-input single-output networks was proposed in~\cite{10061602}, which cooperatively utilizes beamforming and roaming cost information to increase energy efficiency.

However, a common limitation across these studies is the unrealistic assumption that the traffic loads of all SBSs, including those in sleep mode, are perfectly known at each time slot.  
Some studies have used traffic prediction for BS operation and energy management. 
For instance, a traffic-prediction-based BS switching strategy was examined in~\cite{9528008} that accounts for environmental variability across heterogeneous BS types to minimize total power consumption, yet it still assumes perfect load knowledge for all SBSs, including those in sleep mode. Several prior works have applied traffic prediction methods using statistical modeling and deep learning, such as~\cite{10506912}, but these approaches are limited to active BSs and fail to handle the unique constraints posed by sleeping SBSs. Broader overviews of machine learning--based traffic prediction methods are available in~\cite{eswa}, offering valuable context but still not tailored for sleeping SBS scenarios. A clustering-driven method for traffic load prediction that groups BSs and applies deep recurrent neural networks~(RNNs) shows improved accuracy but remains focused on active load forecasting~\cite{jsan9040053}. Overall, a critical gap persists: sleeping-SBS load-estimation errors are not accounted for in cell-switching optimization.

Parallel to cell switching literature, a growing body of work has focused on integrating renewable energy sources into wireless networks to promote sustainability in 6G systems. Specifically, leveraging solar-powered SBSs has the potential to significantly reduce grid energy consumption and carbon emissions. Several works explored energy harvesting and storage-aware management of BSs, proposing adaptive resource allocation strategies based on solar availability. For instance,~\cite{8502042} modeled a joint planning of solar installations and dynamic operation of BSs (sleep-mode scheduling), demonstrating the importance of co-optimizing renewable deployment and dynamic switching in cellular networks. In~\cite{6874568}, the authors proposed an innovative framework for energy cooperation between renewable-powered BSs, where BSs with excess harvested energy can share it with others facing energy shortages.
While these efforts address renewable-aware resource management at the BS level, most treat cell switching and renewable dynamics as loosely coupled problems rather than a unified decision process.

Non-terrestrial elements, notably high-altitude platform stations (HAPS) as a super MBS (SMBS), have been advocated to complement terrestrial infrastructure and improve energy sustainability by enabling vertical offloading and wide area coverage~\cite{9773096,10938203}. 
Yet prior non-terrestrial networks (NTN) studies do not couple per-slot cell switching decisions with explicit estimation of sleeping-SBS loads and renewable-aware switching rules tailored to solar-capable SBSs.

To the best of our knowledge, the current literature lacks a comprehensive framework that jointly considers the dual challenges of traffic load estimation for sleeping SBSs and intelligent cell switching with renewable energy-aware decision-making. Bridging this gap is essential for enabling realistic, sustainable energy management strategies in 6G vertical HetNets (vHetNets), where multiple tiers of terrestrial and aerial BSs coexist with diverse energy capabilities. This paper fills this critical gap by proposing a two-phase framework that first estimates the load of sleeping SBSs and then optimizes cell switching decisions with explicit consideration of solar energy availability across SBSs.

\subsection{Contributions}

We propose a two-phase framework for energy-efficient cell switching in future sustainable vHetNets, comprising MBSs, SBSs, and HAPS-SMBS. Phase~I estimates the loads of \emph{sleeping} SBSs, enabling a feasible and practical implementation of cell switching operations; Phase~II performs \emph{renewable-aware} switching using the most accurate Phase~I estimates. 
The main contributions of this work are summarized as follows:

\begin{itemize}
    \item \textbf{Problem formulation under unobservability:} We consider cell switching when loads of sleeping SBSs are unobserved at the decision horizon and derive a corresponding analytical model for expected power consumption, capturing the impact of estimation errors. We prove that misestimation can flip the optimal ON/OFF decisions and increase total power, quantifying the risk pathways that undermine energy savings or degrade user QoS.

    \item \textbf{Phase I-Load estimation portfolio and analysis:} We develop three estimators spanning different levels of data availability: (i) a distance-based spatial estimator (no historical data), (ii) a multi-level clustering (MLC) estimator using repeated $k$-means with elbow-based cluster selection (limited historical data), and (iii) a long short-term memory (LSTM)-based temporal estimator (full historical data). This portfolio offers strategic flexibility: depending on the richness of available data, the framework can dynamically switch to the most suitable estimator. Such adaptability ensures practical applicability across diverse deployment scenarios, from data-scarce environments to data-rich networks, thereby making the framework more powerful and resilient.

\item \textbf{Phase II-Renewable-aware cell switching with solar-capable SBSs:} Using the most accurate Phase~I estimates, we optimize SBS ON/OFF decisions while explicitly modeling photovoltaic harvesting and storage. We design three operational scenarios (full inclusion, exclusion, and threshold-based inclusion) that reflect realistic hybrid grid–renewable policies. This multi-scenario design provides strategic flexibility, enabling the framework to adapt its switching strategy to varying levels of solar penetration and storage availability, thereby reducing the effective search space while maintaining energy efficiency.

   \item \textbf{HAPS-assisted vHetNet context:} We integrate a HAPS-SMBS tier to enable vertical offloading and wide-area coverage, examining interactions between aerial-terrestrial layering, sleeping-SBS estimation, and renewable-aware switching.

    \item \textbf{Quantitative results on real data:} Using a Milan call detail record (CDR) dataset, Phase~I’s LSTM estimator achieves <1\% mean absolute percentage error (MAPE). In Phase~II, the threshold-based solar integration scenario achieves up to $\sim$23\% NES relative to conventional cell switching under comparable conditions.

 \end{itemize}

\section{System Model}\label{sec:model}
\subsection{Network Model}
We consider a vHetNet comprising a total of $b \in \mathbb{N}$ BSs, indexed by $i \in \mathcal{I}$, where $\mathcal{I} = \{1, 2, \ldots, b\}$. Specifically, we study a macro cell (MC) with one MBS and $s \in \mathbb{N}$ SBSs, indexed by $j\in\mathcal{J}=\{1,2,\ldots,s\}$.
Additionally, a HAPS-SMBS is integrated, potentially serving multiple MCs, into the network.
The primary function of SBSs is to deliver data services and address user-specific requirements, while MBS and HAPS-SMBS ensure consistent network coverage and manage control signals.
A key role of HAPS-SMBS is to efficiently manage traffic offloading from SBSs during low-traffic periods, utilizing its extensive line-of-sight~(LoS) and large capacity~\cite{9380673}. 
This capability not only enhances network flexibility and optimizes capacity utilization but also offers a cost-effective alternative to multiple terrestrial BSs, particularly in fluctuating and high-demand scenarios. 

\subsection{Motivation for a HAPS Tier}

Operating in the stratosphere at altitudes of 20–50 km, HAPS nodes are quasi-stationary and typically confined to a defined cylindrical region~\cite{9380673}. According to ITU recommendations, a HAPS can provide a wide footprint of up to 500 km in radius~\cite{itu_2000}, enabling broad coverage and reinforcing its role as an SMBS for vHetNets. Leveraging this high-altitude position, HAPS-SMBS provides predominantly LoS links and can form dynamic spot beams over a large coverage area, reducing blockage and handovers and making the platform effective for traffic offloading and coverage extension. In addition, HAPS-SMBS can support backhauling (RF/FSO) and operate in higher-frequency bands (e.g., mmWave/THz) where wide bandwidths are available.

Fig.~\ref{HAPS} illustrates the capacity headroom using Shannon’s law \(C=m\,B\log_2(1+\gamma)\), where $C$ is the channel capacity, $B$ is the system bandwidth, $\gamma$ is the signal-to-interference-plus-noise ratio (SINR), and $m$ is a scaling factor representing the effective bandwidth allocation relative to the baseline. The dashed curve is a terrestrial baseline with \(m=1\) at the reference SINR; the x-axis shows SINR degradation to reflect the longer HAPS link distance; the factor \(m>1\) represents additional bandwidth that a HAPS-SMBS can allocate~\cite{9380673}. Because capacity scales approximately linearly with \(m\) and \(B\) but only logarithmically with SINR, modest increases in \(m\) keep HAPS throughput at or above the terrestrial baseline even under notable SINR loss. In practice, a HAPS-SMBS can employ larger values of \(B\), further increasing offloading headroom. This is a physics-based illustration to motivate including a HAPS tier in the vHetNet.

%%%%%%%%%%%%%%%%%%%%%%%%%

\begin{figure}[t]
\centerline{\includegraphics[width = 8.cm ]{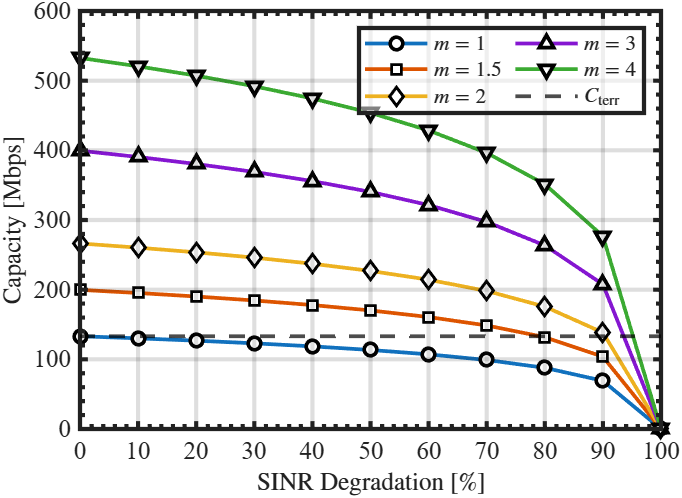}}
\caption{Capacity via \(C=mB\log_2(1+\gamma)\) versus SINR degradation relative to a terrestrial baseline (dashed); \(B=20\) MHz, baseline \(\gamma_0=20\) dB. Curves for \(m\in\{1,1.5,2,3,4\}\) (interpreted as bandwidth scaling). Larger \(m\) offsets SINR loss, providing additional capacity for HAPS offloading.}
%\vspace{-.5cm} 
\label{HAPS}
\end{figure}

\subsection{Network Power Consumption}
The power consumption of each BS in the network is calculated based on the energy-aware radio and network technologies (EARTH) power consumption model~\cite{6056691}. 
For the $i$-th BS, the power consumption at time slot $t$ is expressed as~\cite{7925662}

\begin{equation}
P_i(t)=
\begin{cases}
P_{\mathrm{o},i}+\eta_i\,\lambda_{i,t}\,P_{\mathrm{t},i}, & 0<\lambda_{i,t}\le 1,\\[2pt]
P_{\mathrm{s},i}, & \lambda_{i,t}=0,
\end{cases}
\label{eq1}
\end{equation}
where $P_{\text{o},i}$ represents the operational circuit power consumption, $\eta_i$ is the power amplifier efficiency, $\lambda_{i,t}$ is the load factor (i.e., the normalized traffic load), $P_{\text{t},i}$ is the transmit power, and $P_{\text{s},i}$ is the power consumption in sleep mode. 
The total instantaneous power consumption of vHetNet, denoted as $P_\text{T}$, is given by
%\vspace{-1mm}
\begin{equation}
P_\text{T}(t) = \sum_{i=1}^{b}P_i(t)=P_\text{H}(t)  + P_\text{M}(t)  + \sum\limits_{j = 1}^s {P_j(t) }, 
\label{eq2}
\end{equation}
where $P_\text{H}(t)$ and $P_\text{M}(t)$ denote the power consumption of HAPS-SMBS and MBS at any given moment, respectively, which are calculated based on the $(0 < \lambda _{i,t}  < 1)$ case in \eqref{eq1} as HAPS-SMBS and MBS are always active in our modeling. 
Meanwhile, $P_j(t)$ represents the power consumption of the $j$-th SBS and $s$ signifies the total number of SBSs within the network.

\subsection{Renewable Energy Model}\label{subsec:renewable_model}
We incorporate a subset of SBSs that harvest solar energy and store it in a local battery. Since storage operates on energy, we model all harvested/used/stored quantities on a per-time-slot basis. Let $\Delta t_h=\Delta t/60$ denote the slot duration in hours. For SBS $j$ in slot $t$, the SBS’s energy demand is $E_j(t)=P_j(t)\,\Delta t_h$, where $P_j(t)$ comes from~\eqref{eq1}. The harvested solar energy is
\begin{equation}
E_{H,j}(t)=
\begin{cases}
\zeta\,C_s\,\Delta t_h\,\alpha(t), & t_{\mathrm{p},1}\le t \le t_{p,2},\\[4pt]
0, & \text{otherwise},
\end{cases}
\label{eq:harvest_energy}
\end{equation}
where $\zeta$ is solar conversion efficiency, $C_s$ (kWh/h) is the solar capacity (maximum harvestable energy per hour under peak clear-sky conditions). The factor $\alpha(t)\in[0,1]$ captures time-varying solar availability (e.g., cloud cover, seasonal effects, and intraday irradiance variations), with $\alpha(t)\approx 1$ under clear-sky peak conditions and smaller values under reduced irradiance. The parameters $t_\mathrm{p,1}$ and $t_\mathrm{p,2}$ denote the start and end times of the peak solar period, respectively.

Let $S_j(t)$ be the battery state (stored energy) at the end of slot $t$, with capacity $S_{\max,j}$. 
The renewable energy actually used in slot $t$ is
\begin{equation}
E_{\mathrm{R},j}(t)=\min\big\{E_j(t),\,S_j(t-1)+E_{\mathrm{H},j}(t)\big\},
\label{eq:renew_use}
\end{equation}
and the corresponding grid energy draw is
\begin{equation}
E_{\mathrm{G},j}(t)=E_j(t)-E_{\mathrm{R},j}(t).
\label{eq:grid_energy}
\end{equation}
The storage dynamics follow
\begin{equation}
\begin{aligned}
S_j(t)=\min\Big\{S_{\max,j},\,S_j(t-1)+E_{\mathrm{H},j}(t)-E_{\mathrm{R},j}(t)\Big\},\\
\qquad 0\le S_j(t)\le S_{\max,j}.
\label{eq:storage_energy}
\end{aligned}
\end{equation}

With this formulation, reductions in solar availability directly reduce $E_{\mathrm{H},j}(t)$, and any unmet demand is automatically supplied by the grid via~\eqref{eq:grid_energy}, while the battery state adapts through~\eqref{eq:storage_energy}, which also enforces the storage bounds $0\le S_j(t)\le S_{\max,j}$ (saturation at $S_{\max,j}$ and possible depletion toward $0$). For non-solar SBSs, $E_{\mathrm{H},j}(t)=0$ and $S_{\max,j}=0$, which implies $E_{\mathrm{R},j}(t)=0$ and $E_{\mathrm{G},j}(t)=E_j(t)$.

\begin{figure*}[t]
\centerline{\includegraphics[width = .7\textwidth ]{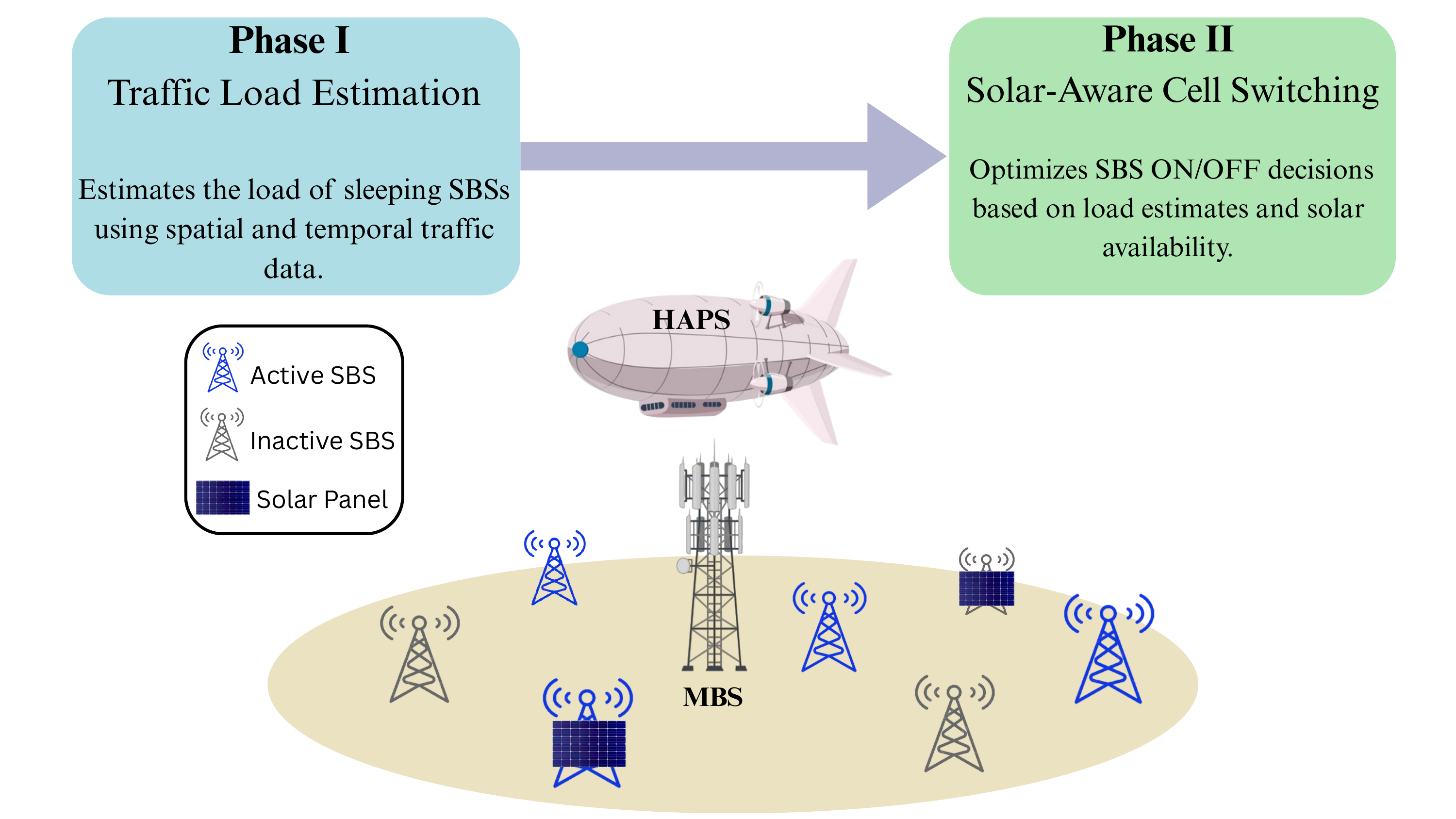}}
\caption{The proposed two-phase optimization framework for a vHetNet. Phase I estimates the traffic load of sleeping SBSs using spatial and temporal data. Phase II uses these estimates along with solar availability to optimize SBS ON/OFF switching, minimizing network power consumption. The system model shows the vHetNet topology with an MBS, multiple SBSs (some solar-powered), and a HAPS enabling vertical offloading.}
\label{fig-1}
%\vspace{-.5cm} 
\end{figure*}
\section{Two-Phase Framework for Power Consumption Optimization}
In this section, we introduce a two-phase framework designed to minimize power consumption in vHetNets by dynamically managing the operational states of SBSs. The framework incorporates both intelligent traffic load estimation and renewable energy awareness to support energy-efficient decision-making across the network.
\subsection{Two-Phase Framework Overview}

The first phase estimates the traffic load of sleeping SBSs using both spatial and temporal methods. These estimations are then utilized in the second phase, which determines the optimal switching states of SBSs to minimize overall power consumption.
A key feature of the second phase—representing a major contribution of this work—is the integration of solar-powered SBSs equipped with energy storage. By jointly considering the estimated traffic load and the availability of harvested solar energy, this phase enables more sustainable and adaptive cell switching decisions. This improves energy efficiency not only through intelligent traffic estimation but also by utilizing renewable energy sources to reduce dependency on the power grid.
Fig.~\ref{fig-1} illustrates the complete framework, combining the vHetNet architecture (including MBS, SBSs, and HAPS-SMBS) with the two-phase optimization process. While the phases are interdependent, they remain modular: Phase I provides the required traffic estimates, and Phase II builds on these to achieve renewable-aware, energy-efficient operation.

\subsection{Cell Switching Problem Formulation under Imperfect Load Estimation}\label{sec:problem}

Our objective is to minimize the total power consumption of the vHetNet, \(P_\text{T}\), by selecting the most energy-efficient switching state. Let
\(\boldsymbol{\Delta}_t = [\delta_{1,t},\delta_{2,t},\ldots,\delta_{s,t}]^\top\) denote the SBS ON/OFF vector at time slot \(t\), where \(\delta_{j,t}\in\{0,1\}\) is the state of the \(j\)-th SBS (0: OFF, 1: ON). The MBS and HAPS-SMBS are always active, i.e., \(\delta_{\text{M},t}=\delta_{\text{H},t}=1\) \(\forall t\).
If the \(j\)-th SBS is switched OFF at time slot \(t\) (i.e., \(\delta_{j,t-1}=1\) and \(\delta_{j,t}=0\)), its load becomes zero, \(\lambda_{j,t}=0\), and its previous-slot load is offloaded to the \emph{macro-tier BSs} \(k\in\{\text{M},\text{H}\}\) according to

\begin{equation}
\lambda _{k,t}  = \lambda _{k,(t - 1)}  + \phi _{j,k} \lambda _{j,(t - 1)}, \quad k=\{\text{M},~ \text{H}\}, 
\label{eq3}
\end{equation}
where \(\lambda_{k,t}\) is the load factor of the MBS (\(k=\text{M}\)) or HAPS-SMBS (\(k=\text{H}\)), and \(\phi_{j,k}=\frac{C_j}{C_k}\) is the relative capacity ratio between the \(j\)-th SBS and macro-tier \(k\). 

Conversely, if the \(j\)-th SBS is switched on at time slot \(t\) (i.e., \(\delta_{j,t-1}=0\) and \(\delta_{j,t}=1\)), it receives normalized load
\(\lambda_{j,t}=\frac{\tau_{j,t}}{C_j}\), where \(\tau_{j,t}\) is its (non-normalized) traffic load\footnote{Traffic load, $\tau$, represents the load of a BS, i.e., the amount of bandwidth occupied, while load factor, $\lambda$, denotes the normalized traffic load, which is computed by dividing the traffic load, $\tau$, to the capacity of the BS, $C$.
Since one can be computed from another, they can be used interchangeably.}, and the macro-tier offload accordingly
\begin{equation}
\lambda_{k,t}=\lambda_{k,t-1}-\phi_{j,k}\,\lambda_{j,t},\quad k\in\{\text{M},\text{H}\}.
\label{eq5}
\end{equation}

Accordingly, the generic cell switching optimization problem can be formulated as
 \begin{IEEEeqnarray*}{lcl}\label{eq:P1}
    &\underset{\mathbf{\Delta_t}}{\text{minimize}}\,\, & ~P_\text{T}(\mathbf{\boldsymbol{\Delta_t}} ) \,  \IEEEyesnumber \IEEEyessubnumber* \label{eq:P1_Obj}\\
    &\text{s.t.} & {\lambda _{\text{M},t}}{\le 1,} \;\\&&{\lambda _{\text{H},t}}{\le 1,}\label{eq:P1_const1}\\
        && {{\delta _{j,t} \in \{ 0,1\}, }} \quad {{\; j = 1,...,s}}, \label{eq:P1_const3}\\
    && {\eqref{eq3},\eqref{eq5}}. \label{eq:P1_const2}
\end{IEEEeqnarray*}
Problem~\eqref{eq:P1} is a mixed-integer linear program (MILP), since it optimizes over binary SBS switching variables $\delta_{j,t}\in\{0,1\}$ under linear constraints.
This generic cell switching optimization model will be customized with the renewable energy integration later in the paper.

Substituting \eqref{eq1} into \eqref{eq2} and rearranging, we obtain the closed-form expression below. 
In particular, for each SBS $j$, the contribution is written as
$\big(P_{\mathrm{o},j}+\eta_j\lambda_{j,t}P_{\mathrm{t},j}\big)\delta_{j,t}+P_{\mathrm{s},j}(1-\delta_{j,t})$,
so that $\delta_{j,t}=1$ selects the active-mode power and $\delta_{j,t}=0$ selects the sleep-mode power.
Since the MBS and HAPS-SMBS are always active (i.e., $\delta_{\mathrm{M},t}=\delta_{\mathrm{H},t}=1$ for all $t$), their expressions reduce to the active-mode term in \eqref{eq1}.

\begin{equation}
\begin{aligned}
P_\text{T} (\mathbf{\boldsymbol{\Delta_t}} )
 &={(P_{\text{o,H}}  + \eta _\text{H} \lambda _{\text{H},t} } P_{\text{t,H}} )+{(P_{\text{o,M}}  + \eta _\text{M} \lambda _\text{{M},t} } P_{\text{t,M}})\\
 &\hspace{-2em}\;\;\;\;\;\;\;+\Bigg(\sum\limits_{j = 1}^s {(P_{\text{o},j}  + \eta _j \lambda _{j,t}  P_{\text{t},j} )\delta _{j,t} +P_{\text{s},j}(1-\delta _{j,t})\Bigg)}.
\end{aligned}
\label{eq7}
\end{equation}
In the constraints \eqref{eq:P1_const1} and \eqref{eq:P1_const2}, $\lambda_\text{M}=[0,1]$ and $\lambda_\text{H}=[0,1]$ are defined as positive real numbers, $\mathbb{R}^+$, ensuring that the operational capacities of both the MBS and HAPS-SMBS are never exceeded, thereby upholding the QoS requirements. 
Importantly, the optimal state vector, ${\mathbf{\boldsymbol{\Delta}} _{\text{opt}}}$, minimizes ${P_\text{T}(\mathbf{\boldsymbol{\Delta}} )}$, which is a function of the load factors, ${\lambda }$, of SBSs. 
It should be noted that the optimization model in \eqref{eq:P1} presumes complete knowledge of all BSs' traffic loads, including those in sleep mode. This presumption, common in the literature~\cite{maryam,HETS2023P,ELAA2022JR,EOMK2017JR, Metin_VFA_CellSwitch,11352980}, poses a pivotal question: How can we accurately determine the traffic load for a sleeping SBS to decide its next state? The lack of precise knowledge about ${\lambda }$ for these SBSs necessitates a reliable estimation method. Addressing this estimation challenge is essential for refining our optimization strategy, which we will explore in the subsequent section.
\subsubsection{Power Consumption with the Erroneous Load Estimations}
% Proceeding to power consumption estimation, we address the imperative of approximating the traffic load for sleeping SBSs. 
Due to the unavailability of real-time traffic load data for the sleeping SBSs, we consider an estimated load, ${\hat \lambda}$, introducing a layer of uncertainty in the optimization process. Consequently, we redefine the network's total power consumption given in~\eqref{eq7} with the estimated load factors, $\hat \lambda$, as
\begin{equation}
\begin{aligned}
P_\text{est} (\mathbf{\boldsymbol{\Delta}} )
 &={(P_{\text{o,H}}  + \eta _\text{H} \hat \lambda _\text{H} } P_{\text{t,H}} )+{(P_{\text{o,M}}  + \eta _\text{M} \hat \lambda _\text{M} } P_{\text{t,M}})\\
 &\hspace{-2em}\;\;\;\;\;\;\;+\sum\limits_{j = 1}^s \Bigg({(P_{\text{o},j}  + \eta _j \hat \lambda _j  P_{\text{t},j} )\delta _j +P_{\text{s},j}(1-\delta _j)\Bigg)}.
\end{aligned}
\label{eq8}
\end{equation}
Depending on the accuracy of our estimation, this could be either perfect or imperfect estimation. In the case of perfect estimation, where ${\lambda _j}={\hat \lambda _j},~\forall j$, the estimated power consumption, $P_\text{est}$, accurately captures the ground truth, $P_\text{T}$. 
Achieving perfect estimation is the goal of our work and should be the goal of any other work employing a cell switching scheme. 
% However, this need for estimation has not already been introduced in the literature. To the best of our knowledge, this is the first time we are introducing and addressing this gap between perfect knowledge and reality assumption. 

In the case of imperfect estimation, where $\exists j \text{ such that } \lambda_j \ne \hat{\lambda}_j$, there would be a probability of error, ${p_\text{err}}$, and we can define an expected value for the total power consumption of the network as
\begin{equation}  
E[P] = P_\text{est}.p_\text{err}  + P_\text{T}.(1 - p_\text{err}).
\label{eq9}
\end{equation}
% Considering the need for accurate traffic load estimation, the following theorem investigates two scenarios that illustrate the potential consequences of estimation inaccuracies on the network's power consumption optimization.
\begin{theorem}
Errors in estimating the traffic load of SBSs in a vHetNet can lead to changes in the optimal state vector, thereby affecting the total power consumption of the network.
\end{theorem}
\begin{proof}
We analyze the impact of traffic load estimation errors on network power consumption through two primary scenarios related to the operational dynamics of SBSs.
\begin{itemize} 
\item \underline{Overestimation Scenario}: If the estimated load ${\hat \lambda _j}$ at time $t_1$ for a sleeping SBS is more than the actual load and a set threshold $\lambda _\text{th}$, it might mistakenly trigger the SBS's transition to an active state, leading to unnecessary power consumption. The probability of this erroneous transition, $p_\text{err} ^{\text{off} \to \text{on}}$, is defined as
\begin{equation}
p_\text{err} ^{\text{off} \to \text{on}}  = \text{Pr}\{ \hat \lambda _j  > \lambda _\text{th}\mid \lambda _j  \le \lambda _\text{th} \}.
\label{eq10}
\end{equation}

Note that $p_\text{err}^{\text{off}\to\text{on}} \in [0,1]$ is a conditional probability (a real-valued scalar) determined by the joint statistics of the true and estimated loads $(\lambda_j,\hat{\lambda}_j)$.
The expected value of the associated power consumption error, $P_\text{err}$, due to this transition is calculated as
\begin{equation} 
\begin{array}{l}
E[P_\text{err}]  = \big[(\eta _\text{H} \phi _{j,\text{H}} \lambda _j P_\text{t,H}  + P_{\text{s},j} ) - (P_{\text{o},j} \\ + \eta _j \hat \lambda _j P_{\text{t},j} )\big]\times p_\text{err} ^{\text{off} \to \text{on}}. 
%\.pr\{ \lambda _{i,t_1 }  \le \lambda _\text{th} \} .pr\{ \lambda _{i,t_0}  \le \lambda _\text{th} \} 
\end{array}
\label{eq11}
\end{equation}
\item \underline{Underestimation Scenario}: Conversely, if the actual load of an SBS at time $t_1$ is higher than the estimated and it mistakenly remains in sleep mode, this can lead to insufficient traffic management. 
The probability of such underestimation, $p_\text{err} ^{\text{on} \to \text{off}}$, is defined as
\begin{equation}
p_\text{err} ^{\text{on} \to \text{off}}  = \text{Pr}\{\hat \lambda _j  < \lambda _\text{th}\mid \lambda _j  \ge \lambda _\text{th} \}.
\label{eq12}
\end{equation}
Similarly, $p_\text{err}^{\text{on}\to\text{off}} \in [0,1]$ is a conditional probability capturing the likelihood of missed activation when the true load exceeds $\lambda_\text{th}$ but the estimate falls below it.
The expected value of the error in the total power consumption of the network, $P_\text{err}$, for this underestimation scenario is given by
\begin{equation} 
\begin{array}{l}
E[P_\text{err}]  = \big[(P_{\text{o},j}  + \eta _j \hat \lambda _j P_{\text{t},j} )-(\eta _\text{H} \phi _{j,\text{H}} \lambda _j P_\text{t,H} \\ + P_{\text{s},j} )\big]\times p_\text{err} ^{\text{on} \to \text{off}}.%\ .pr\{ \lambda _{i,t_1 }  \ge \lambda _\text{th} \} .pr\{ \lambda _{i,t_0}  \le \lambda _\text{th} \} 
\end{array}
\end{equation}
\label{eq13}
\end{itemize}
These scenarios illustrate the significant impact that accurate load estimation of sleeping SBSs has on the efficient optimization of network power consumption, affirming the core assertion of the Theorem.
\end{proof}
\subsection{Phase I: Multi-Dimensional Traffic Load Estimation-Spatial and Temporal Perspectives}\label{sec:method}

\begin{figure*}[t]
\centerline{\includegraphics[width = .7\textwidth ]{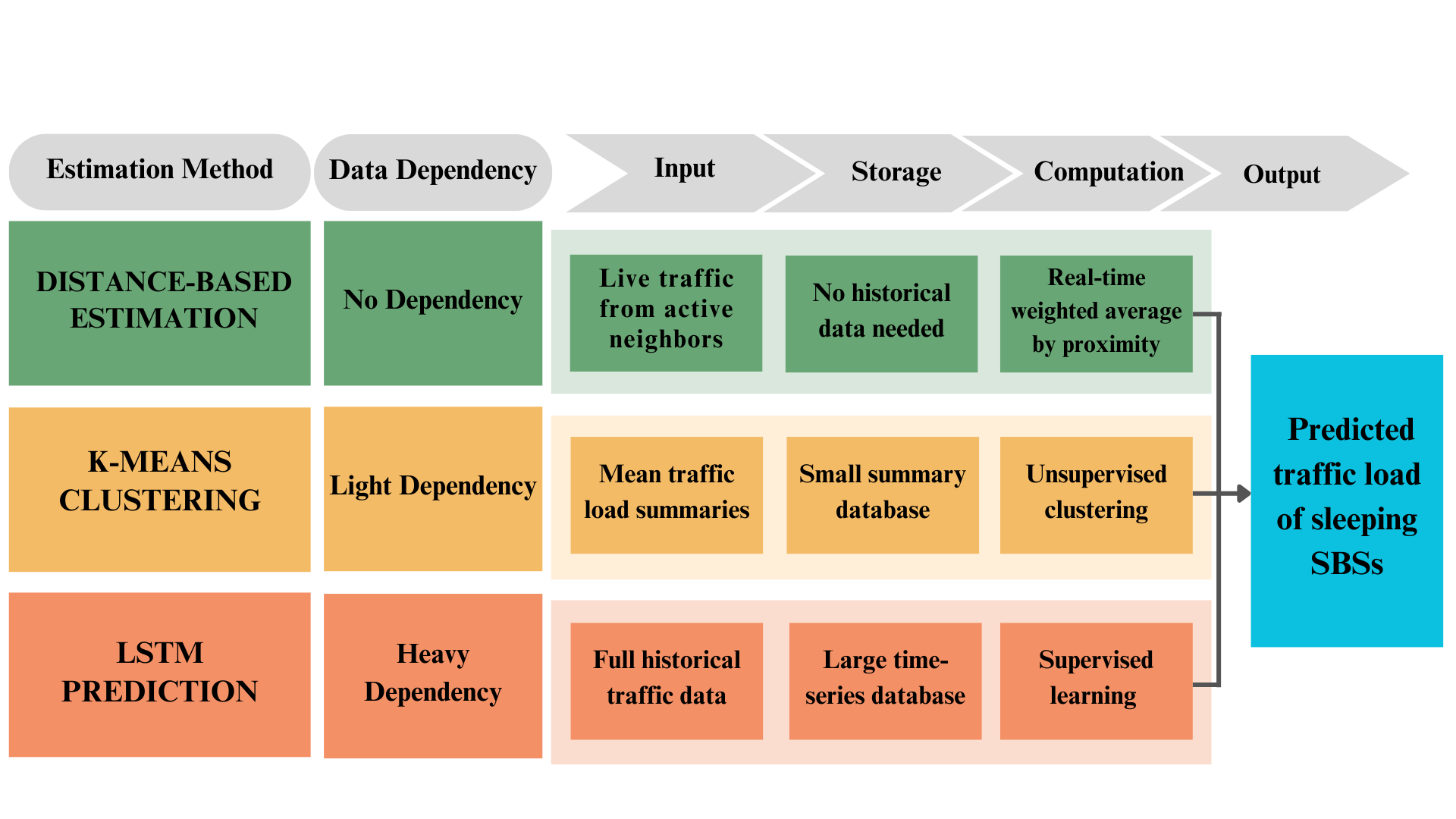}}
\caption{Taxonomy of Phase~I traffic load estimation methods, categorized by historical-data dependency (no/light/heavy) and the associated input, storage, and computational requirements. Any of the three can be used depending on data availability and deployment constraints.}
\label{fig-2}
%\vspace{-.5cm} 
\end{figure*}

Accurate estimation of the traffic load for sleeping SBSs is critical for effective and scalable cell switching in vHetNets. The choice of estimation method depends heavily on the availability of historical data and the trade-offs between estimation accuracy and implementation complexity.
In this work, we propose and classify traffic load estimation methods into three distinct categories based on their level of dependency on historical data, as illustrated in Fig.~\ref{fig-2}. Importantly, Fig.~\ref{fig-2} provides a categorization of candidate estimators rather than a selection rule: depending on the deployment’s data availability and complexity constraints, Phase~I can employ any one of these methods. In Section~\ref{sec:performance}, we evaluate all three estimators under the same dataset and report their comparative accuracy and impact on Phase~II switching.
This classification highlights the trade-off between data availability, computational complexity, and estimation accuracy, and serves as a practical guide for different deployment scenarios:

%%%%%%%%%%%%%%%%%%%
\begin{itemize}
    \item \textbf{No Historic Data Dependency – Distance-Based Method:} A spatial method that does not require any historical data. It estimates the traffic load of a sleeping SBS based on the real-time traffic loads of its active neighboring SBSs.
    
    \item \textbf{Light Historic Data Dependency – Clustering-Based Method:} Also spatial, this method utilizes spatial correlations and uses only summary statistics (e.g., average traffic load over time) from neighboring SBSs. It requires limited historical data and offers a balance between complexity and performance.
    
    \item \textbf{Heavy Historic Data Dependency – LSTM-Based Method:} A temporal approach that utilizes the full historical traffic sequence for each SBS to train a predictive model. It is designed for high-accuracy estimation when abundant historical data is available.
\end{itemize}

This categorization reflects a novel contribution of our work by bridging the gap between estimation accuracy and practical feasibility in real-world vHetNet deployments.

\subsubsection{Geographical Distance-Based Traffic Load Estimation}\label{dist}
This method estimates the traffic load of a sleeping SBS based on the proximity of its neighboring cells. By applying a distance-based weighting scheme, the influence of each neighboring cell is adjusted according to its geographical proximity, prioritizing nearby cells in the estimation process. Since this approach only uses current information from neighboring active cells and does not require historical traffic data, it belongs to the \textit{no dependency} category described earlier.

The estimated traffic load of the $j$-th sleeping SBS, $\hat\lambda_j$, is computed as

\begin{equation}
\hat \lambda _j  = {\sum\limits_{a = 1}^N {\lambda _a  \times w_{j,a} } }\Big/{{\sum\limits_{a = 1}^N {w_{j,a} } }},
\label{eq15}
\end{equation}
where $\lambda_a$ denotes the traffic load of the $a$-th neighboring SBS, and $w_{j,a}$ is the corresponding weighting factor. The number of neighboring SBSs included in the estimation is denoted by $N$. Note that usually $N\gg s$, as $N$ encompasses all the cells available for the estimation process, while $s$ is the number of SBSs within a single VHetNet with a single MC.

The weighting factor $w_{j,a}$ is defined as
\begin{equation} 
w_{j,a}  = \frac{{d_{\max } }}{{d_{j,a}^n }},\quad
n \in \mathbb{R}^+, 
\label{eq16}
\end{equation}
where ${d_{\max }}$ is the maximum distance between the sleeping SBS and its neighboring cells included in the estimation, and ${d_{j,a}}$ is the distance between the sleeping SBS ${j}$ and the neighboring SBS ${a}$.

\begin{algorithm}[H]
\caption{Two-Phase Cell Switching Framework} 
\DontPrintSemicolon
\SetAlgoLined

\KwIn{
$\tau_{j,t}$, $C_j$, $\Delta t$,
estimator choice $\mathcal{E}\in\{\mathrm{Dist},\,\mathrm{MLC},\,\mathrm{LSTM}\}$,
$S_{\mathrm{r}}\subseteq\mathcal{J}$ and renewable-model parameters,
power-model parameters in \eqref{eq1}--\eqref{eq2};
scenario/threshold parameter $\gamma$ (if applicable)}

\KwOut{
$\boldsymbol{\Delta}_t$, $P_{\mathrm{T}}(t)$}

\BlankLine
Initialize loads $\lambda_{j,t}=\tau_{j,t}/C_j$ and renewable storage states\;

\For{each time interval $\Delta t$ }{
    \textbf{Phase I (Sleeping-cell load estimation):}\;
    Construct $\hat{\lambda}_{j,t}$ for sleeping SBSs using the selected estimator $\mathcal{E}$:\;
    \uIf{$\mathcal{E}=\mathrm{Dist}$}{
        Obtain $\hat{\lambda}_{j,t}$ from \eqref{eq15}-\eqref{eq16} equations in Sec.~\ref{dist}\;
    }
    \uElseIf{$\mathcal{E}=\mathrm{MLC}$}{
        Obtain $\hat{\lambda}_{j,t}$ via multi-level clustering using Algorithm~\ref{MLCALG}\;
    }
   \uElseIf{$\mathcal{E}=\mathrm{LSTM}$}{ 
        Obtain $\hat{\lambda}_{j,t}$ by forward inference of the trained LSTM (Sec.~\ref{lstmapp})\;
    }

    \BlankLine
    \textbf{Phase II (Renewable-aware switching optimization):}\;
    Form the per-slot optimization (Problem~\eqref{eq:P1}, later extended with renewables) using measured/estimated loads
    (i.e., use $\lambda_{j,t}$ for active SBSs and $\hat{\lambda}_{j,t}$ for sleeping SBSs)\;
    Solve for the SBS state vector $\boldsymbol{\Delta}_t=[\delta_{1,t},\ldots,\delta_{s,t}]^\top$ using the selected switching policy
    (e.g., exhaustive search benchmark with the chosen renewable scenario/threshold $\gamma$)\;
    
    Update macro-tier loads using \eqref{eq3}--\eqref{eq5}\;
    Update renewable states for $j\in S_{\mathrm{r}}$ using \eqref{eq:harvest_energy}--\eqref{eq:storage_energy}\;
    Compute and record $P_{\mathrm{T}}(t)$ (e.g., via the closed form in \eqref{eq7} / its renewable-aware counterpart)\;
}
\label{alg:two_phase_sim_notation}
\end{algorithm}

\subsubsection{Clustering-Based Traffic Load Estimation}

This method relies on clustering SBSs according to their traffic patterns and estimating the traffic load of a sleeping SBS based on the average load of active SBSs within the same cluster. Since it requires only summary statistics of historical data—such as the mean traffic load over a given period—it falls under the \textit{light dependency} category in our data classification framework. This makes it suitable for scenarios where full-time series data is unavailable, yet basic statistical profiles of traffic are retained.

We employ the $k$-means algorithm, an unsupervised machine learning technique, to cluster the SBSs. The $k$-means algorithm partitions the SBSs into $G$ clusters (i.e., $k=G$) by iteratively minimizing the within-cluster sum of squared distances to the cluster centroids. Starting from an initial set of centroids, it alternates between (i) assigning each SBS to the nearest centroid (assignment step) and (ii) updating each centroid as the mean of the SBSs assigned to that cluster (update step), until convergence. This yields clusters whose members have similar traffic-profile features (in the Euclidean sense) and provides a simple, scalable way to form behaviorally similar SBS groups for load estimation. In practice, $k$-means can be sensitive to centroid initialization; thus, clustering can be repeated with multiple random initializations, retaining the run with the smallest SSE.

The number of clusters, a key hyperparameter in the $k$-means algorithm, is determined using the elbow method~\cite{article}, which evaluates different cluster counts by calculating the sum of squared errors (SSE) between data points and their assigned centroids. The SSE is given by

\begin{equation} 
SSE = \sum_{g=1}^{G}\sum_{x_m \in \kappa_g} \left\|x_m-\varkappa_g\right\|^2,
\label{eq18}
\end{equation}
where $G$ represents the optimal number of clusters, $\varkappa_g$ the centroid of each cluster $\kappa_g$, and $x_m$ each sample in $\kappa_g$. 
The optimal cluster number is identified at the point where the SSE curve forms an "elbow" before flattening.
\paragraph{Multi-level Clustering-Based Traffic Load Estimation}
The MLC approach involves repeated clustering of SBSs based on their traffic patterns to estimate the traffic load of offloaded SBSs. 
At each level, the traffic load of a sleeping SBS is estimated based on the average load of active SBSs in the same cluster. This iterative approach, as outlined in Algorithm~\ref{alg1}, progressively refines clustering with each layer, leading to more precise traffic load estimations for sleeping SBSs. 
We provide pseudocode for MLC since it explicitly iterates across clustering levels, whereas the distance-based estimator is given in closed form, and the LSTM estimator follows a standard train--infer pipeline described in the LSTM subsection.
Unlike the distance-based estimator that implicitly equates geographic proximity with correlation, MLC clusters SBSs by traffic-profile similarity (based on their historical/aggregate behavior), which can better capture sharp spatial heterogeneities (e.g., SBSs located on different sides of a traffic boundary, such as a dense downtown edge versus a nearby low-demand suburban/rural area, or hotspot cells near stadiums/events) where nearby SBSs may be weakly correlated.
Moreover, since MLC is learned from longer-term traffic profiles (light historical dependency), it is generally less sensitive to transient slot-level variations than purely instantaneous distance-based interpolation; in dynamic environments, clusters can be refreshed periodically using recent summary statistics, with fallback to alternative estimators if traffic drifts faster than the refresh interval.

\subsubsection{Temporal Traffic Load Prediction Using LSTM}\label{lstmapp}
To enhance the accuracy of traffic load estimation, we incorporate LSTM networks for temporal load prediction. As this approach relies on full historical traffic sequences to train a predictive model, it falls under the \textit{heavy dependency} category of our data dependency framework. LSTM networks, a variant of RNN, are particularly well-suited for capturing long-term dependencies in sequential data. 
%The LSTM cell structure, shown in Fig.~\ref{fig-3}, consists of three gates (forget, input, and output) and a cell state that acts as memory.
The LSTM cell structure consists of three gates (forget, input, and output) and a cell state that acts as memory.
The governing equations of an LSTM cell are

\begin{subequations}\label{eq:lstm}
\begin{align}
    f_t &= \sigma\!\big(W_f [h_{t-1}, x_t] + \beta_f\big), \label{eq:lstm_forget} \\
    \iota_t &= \sigma\!\big(W_{\iota} [h_{t-1}, x_t] + \beta_{\iota}\big), \label{eq:lstm_input} \\
    \tilde{c}_t &= \tanh\!\big(W_c [h_{t-1}, x_t] + \beta_c\big), \label{eq:lstm_candidate} \\
    c_t &= f_t \odot c_{t-1} + \iota_t \odot \tilde{c}_t, \label{eq:lstm_cell_state} \\
    o_t &= \sigma\!\big(W_o [h_{t-1}, x_t] + \beta_o\big), \label{eq:lstm_output} \\
    h_t &= o_t \odot \tanh(c_t), \label{eq:lstm_hidden}
\end{align}
\end{subequations}

where $f_t$, $\iota_t$, $o_t$, and $c_t$ denote the forget gate, input gate, output gate, and cell state, respectively. The weights $W$ and biases $\beta$ are trainable parameters, $\sigma$ is the sigmoid activation function, and $\odot$ denotes element-wise multiplication.

In these equations, $f_t$ represents the forget gate, which determines the information to discard from the previous cell state. The input gate, denoted as $\iota_t$, decides what new information to store in the cell state, while $\tilde{c}_t$ represents the candidate cell state that introduces this new information. The updated cell state, $c_t$, combines the retained memory from the forget gate with the new information from the input gate. The output gate, $o_t$, decides the final output of the cell, which is represented by the hidden state $h_t$. 
The parameters $W$ and $\beta$ are the trainable weights and biases of the model, respectively. The sigmoid activation function, $\sigma$, is used in the gates to limit their outputs between 0 and 1, effectively acting as a gating mechanism. 
\begin{algorithm}[h!]
\caption{Clustering-based traffic load estimation via Multi-Level Clustering (MLC) using $k$-means
}
\DontPrintSemicolon  % Removes semicolons
\SetAlgoLined  % Enables vertical lines
\KwData{Traffic loads of SBSs $\lambda_{a}$, maximum number of levels $L$}
\KwResult{Clustered SBSs with estimated traffic loads}

\SetKwFunction{FMain}{MLC\_k\_means}
\SetKwProg{Fn}{Procedure}{:}{}
\Fn{\FMain{$\lambda$, $L$}}{
    Determine the optimal number of clusters $G$ using the elbow method\;
    Initialize level index $l = 1$\;
    \While{$l \leq L$}{
        Perform $k$-means clustering on $\lambda$ to form $G$ clusters\;
        \For{cluster $\kappa_g$}{
            Calculate the mean traffic load $\mu_\text{m}$\;
            \For{sleeping SBS in $\kappa_g$}{
                Estimate the traffic load as $\mu_\text{m}$\;
            }
        }
        Update $\lambda$ with estimated ones for sleeping SBSs\;
        Increment the level count $l$ by 1\;
    }
    \Return The final clusters with estimated traffic loads\;
}
\label{MLCALG}
\end{algorithm}

\paragraph{Application of LSTM to Traffic Load Estimation}

In this study, LSTM networks are employed to predict the future traffic loads of SBSs using historical traffic data. The dataset, consisting of traffic loads collected over 30 days with 144 time slots per day, is preprocessed to remove outliers using z-score filtering with a threshold of 2.5. Subsequently, the dataset is shuffled to improve generalization during training.
The preprocessed data is divided into two subsets: 60\% of the data is used for training the model, while the remaining 40\% is reserved for testing. A sliding window mechanism is applied to the data, with a window size of 8 time steps, to create sequences of input-output pairs. Each sequence comprises 8 previous time steps as input and the next time step as the target output. The look-back window length and the number of LSTM units are the two key capacity hyperparameters of the predictor. We selected these values based on an empirical sensitivity analysis (Fig.~5), which evaluates multiple window/unit configurations and supports the final choice as a practical accuracy--complexity trade-off.
The LSTM model is configured with the parameters summarized in Table \ref{table:Phase1}. These include a learning rate of 0.001, one LSTM layer with 10 units, and a dense layer for output. The model is trained over 50 epochs using the mean absolute error (MAE) loss function and the Adam optimizer, with a batch size of 32. After training, the model is evaluated on the test dataset, and the predicted traffic loads are compared with the actual values. 
These parameters are chosen following standard time-series practices and preliminary tuning; in particular, the learning rate ensured stable convergence, the number of units balanced accuracy and overfitting, and the epoch/batch settings provided good generalization at reasonable cost.
From a deployment perspective, the main computational burden of the LSTM-based estimator arises during training, which can be performed offline using historical traffic traces. In contrast, real-time operation only requires lightweight forward inference to predict the next-slot load \cite{8585453}. To accommodate evolving traffic patterns in highly dynamic networks, the model can be retrained periodically at slower time scales (e.g., daily/weekly) while preserving low-latency inference in operation. Moreover, if a newly deployed SBS lacks sufficient historical data, the framework can temporarily fall back to spatial estimators (distance-based or MLC) until adequate data are collected, ensuring robust operation during network evolution.

\subsection{Phase II: Renewable Energy-Aware Cell Switching Optimization}

The second phase of our framework focuses on optimizing the switching states of SBSs to minimize total network power consumption, with an emphasis on integrating renewable energy contributions from solar-powered SBSs. This phase is interdependent on Phase I, as it relies on the traffic load estimates generated earlier to determine which SBSs should remain active and which can be safely (i.e., without causing QoS degradations) switched off to conserve energy. Since Phase~II operates on per-SBS estimated load factors from Phase~I, spatially heterogeneous traffic demand (e.g., hotspots versus low-demand areas) is naturally handled through these inputs and the associated capacity constraints.
In particular, SBSs serving persistent hotspots tend to remain active because switching them off would either violate macro-tier capacity constraints or yield limited net power reduction after offloading.
%%%%%%%%%%%%%%%%%%%%%
To further enhance sustainability and reduce reliance on grid power, we extend the network model by equipping a subset of SBSs with solar panels and batteries. These solar-capable SBSs can harvest and store solar energy for their operations, contributing to a more energy-efficient and environmentally sustainable vHetNet design—an important objective in the context of future 6G networks.

Since the renewable model in Section \ref{subsec:renewable_model} is expressed in per-slot energy, we use the average-power conversions

\begin{equation}
\bar P_{X,j}(t)\triangleq \frac{E_{X,j}(t)}{\Delta t_h},\qquad \Delta t_h=\frac{\Delta t}{60}\;\text{(hours)},
\label{eq:power_conv}
\end{equation}
for $X\in\{\mathrm{H},\mathrm{R},\mathrm{G}\}$. Note that $E_j(t)=P_j(t)\Delta t_h$ implies $\bar P_j(t)=P_j(t)$.

Let $S_{\mathrm{r}}\subseteq\mathcal{J}$ denote the solar-capable SBSs. Based on the extended model incorporating renewable energy, the power consumption optimization problem for Phase II is formulated as
\begin{IEEEeqnarray*}{lcl}\label{eq:P2}
    &\underset{\boldsymbol{\Delta}_t}{\text{minimize}}\,\, & ~P_\mathrm{T}(t)  \,  \IEEEyesnumber \IEEEyessubnumber* \label{eq:P2_Obj}\\ 
    &\text{s.t.} & {\lambda _{\mathrm{M},t}}{\le 1,} \; \\&&{\lambda _{\mathrm{H},t}}{\le 1,} \label{eq:P2_const1}\\
    && {{\delta _{j,t}  \in \{ 0,1\}, }}  {{\;\; j=1,2,...,s }}, \label{eq:P2_const2}\\
    &&  {\eqref{eq3},\eqref{eq5}} \label{eq:P2_const3},\\
    && P_{\mathrm{G},j}(t)=P_j(t),\quad \forall j \notin S_\mathrm{r}\label{eq:P2_const4}\\
    && {\eqref{eq:harvest_energy}-\eqref{eq:storage_energy}}, \quad \forall j \in S_\mathrm{r}. \label{eq:P2_const5}      
\end{IEEEeqnarray*}

The objective function \eqref{eq:P2_Obj} aims to minimize the total network power consumption based on the switching vector $\boldsymbol{\Delta_t}$. Constraints \eqref{eq:P2_const1}-\eqref{eq:P2_const3} ensure macro-tier feasibility and binary SBS decisions via the offloading relations. These feasibility constraints also make the framework robust to extreme traffic peaks: when traffic demand increases, switching additional SBSs to sleep would increase the offloaded load and may violate $\lambda_{\mathrm{M},t}\le 1$ and/or $\lambda_{\mathrm{H},t}\le 1$, and thus such decisions are automatically excluded. Consequently, the optimizer naturally keeps more SBSs active during peak intervals to preserve QoS.
For non-solar SBSs, \eqref{eq:P2_const4} enforces that their entire demand is supplied by the grid. For solar-capable SBSs, harvesting, renewable usage, grid draw, and storage follow the \emph{Renewable Energy Model} in Section \ref{subsec:renewable_model}; we apply \eqref{eq:harvest_energy}–\eqref{eq:storage_energy} and obtain the corresponding powers via \eqref{eq:power_conv}. 
In particular, the renewable usage is bounded by the available harvested-plus-stored energy via \eqref{eq:renew_use}, and any unmet demand is supplied by the grid through \eqref{eq:grid_energy}, ensuring feasibility even under renewable shortage.
Since this renewable model is defined per SBS and per time slot, heterogeneous solar availability across the service area and over time is directly reflected through $E_{\mathrm{H},j}(t)$ and the resulting storage evolution $S_j(t)$, while any renewable shortfall is automatically compensated via the grid-draw relation in \eqref{eq:grid_energy}. Therefore, when storage is depleted and harvesting is low (e.g., $S_j(t\!-\!1)\!\approx\!0$ and $E_{\mathrm{H},j}(t)\!\approx\!0$), the model yields $E_{\mathrm{R},j}(t)\!\to\!0$ and the SBS demand is naturally met by the grid through \eqref{eq:grid_energy}.
Problem~\eqref{eq:P2} can be cast as an MILP in $\delta_{j,t}\in\{0,1\}$, since the objective and constraints are linear/affine in $\boldsymbol{\Delta}_t$ and the renewable-energy relations \eqref{eq:harvest_energy}--\eqref{eq:storage_energy} are piecewise-linear and admit standard linear reformulations.

To address the challenges of minimizing network power consumption while integrating renewable energy, we propose a modified cell switching algorithm that dynamically accounts for the contributions of solar-capable SBSs.
To showcase the proof-of-concept (PoC) for our novel two-phase framework and demonstrate its effectiveness, we implement an optimal algorithm for the cell switching component to avoid diverting the focus from the framework itself to the algorithmic behavior. Accordingly, we employ an exhaustive search (ES) algorithm, which guarantees finding the optimal solution by evaluating all possible ON/OFF state combinations of SBSs. For a network with $s$ SBSs, the algorithm explores $2^s$ configurations, where ${\delta _j \in {0,1}}$ indicates whether the $j$-th SBS is OFF or ON. 
We emphasize that ES is used in this paper as an optimal benchmark to validate the proposed two-phase framework and quantify the performance limits of renewable-aware switching under accurate load estimation. In particular, using ES enables us to isolate the impact of Phase~I estimation accuracy and renewable-awareness in Phase~II without conflating the results with solver-dependent suboptimality, tuning, or convergence behavior. For a network with $s$ candidate SBS decision variables, ES evaluates all ON/OFF configurations and has exponential complexity $O(2^{s})$, which becomes prohibitive in large-scale deployments. Similar scalability challenges are common in large-scale energy-aware optimization systems (e.g.,~\cite{10976414}).
Importantly, the proposed two-phase framework does not depend on a specific optimization solver: Phase~I provides the required sleeping-cell load estimates, and Phase~II can be implemented using alternative low-complexity methods with scalable behavior (e.g., greedy sorting or learning-based policies). Such approaches are generally suboptimal compared to ES, but they offer significantly improved scalability for large-scale deployments (e.g.,~\cite{maryam,Metin_VFA_CellSwitch}). A detailed integration and benchmarking of these scalable solvers is beyond the scope of this paper and is left for future work.
In this paper, we adopt ES to avoid confounding the analysis with solver suboptimality and to isolate the incremental impact of renewable integration on switching decisions and energy savings.
Accordingly, the non-renewable baseline reported in Section~IV uses the same ES-based switching engine with the renewable model disabled (i.e., all SBS demand is supplied by the grid).

The proposed approach selectively includes or excludes solar-capable SBSs from the search space based on their stored energy levels, dynamically adjusting their contribution to the optimization process. Additionally, the algorithm calculates grid power consumption by accounting for the renewable energy contributions of solar-capable SBSs. Specific adjustments, such as threshold-based exclusion, are introduced to tailor the approach to varying operational conditions. To evaluate the proposed solution, we consider three distinct scenarios based on different strategies for utilizing renewable energy, which are described in detail in the following paragraphs.

\section{Performance Evaluation}\label{sec:performance}
This section presents the performance evaluation of our proposed two-phase framework. We assess (i) the accuracy and efficiency of the traffic load estimation methods in Phase I and (ii) the effectiveness of the renewable energy-aware cell switching strategies in Phase II. All simulations are conducted using the Milan dataset described in Section~\ref{sec:dataset}. Relevant simulation parameters are summarized in Table~\ref{table:Phase1} and Table~\ref{table:Phase2}.

\subsection{Simulation Setup and Dataset}\label{sec:dataset}
To assess power consumption as defined in~\eqref{eq1}, we require the load factor $\lambda_i$ for each SBS. 
Thus, we employ a real CDR data set from Telecom Italia~\cite{DVN/EGZHFV_2015} that captures user activity in Milan, partitioned into 10,000 grid cells of $235\times 235$ meters. This activity includes calls, texts, and Internet usage recorded every 10 minutes over November and December 2013. We consolidate these activities into a single measure of traffic load per grid. 
In our simulations, each SBS is mapped to a specific Milan grid cell; therefore, its geographic location and traffic load are fixed by the CDR grid measurements. The neighbor set used by spatial estimators is defined deterministically from the grid geometry (e.g., adjacent/nearby cells), which preserves the dataset’s inherent spatial traffic correlation. The only randomness is the selection of the SBS whose load is treated as unobserved (sleeping) in each trial; we repeat this process over $300$ Monte Carlo realizations and report averaged results.
The Milan CDR dataset provides terrestrial grid-level traffic measurements that are used to assign SBS traffic loads. The HAPS tier is not part of the dataset; it is modeled analytically in the system model and used in the evaluation through the offloading and capacity constraints.

\subsection{Phase I: Evaluation of Traffic Load Estimation Accuracy and Impact}

\begin{centering}
\begin{table}[t] 
\caption{Simulation Settings for Phase I Estimation Methods}
\centering % used for centering table 
\begin{tabular}{|c|c|} % centered columns 
  \hline % inserts single horizontal line
  \text{Parameter} & \text{Value} \\
  \hline
\multicolumn{2}{|c|}{\textbf{Spatial Estimation}} \\
\hline
Number of SBSs & 5000 \\

Number of time slots & 144 \\

Time slot duration & 10 m \\

Number of days & 30 \\

Number of iterations & 300 \\

Optimal $G$ using elbow method & 3 \\
\hline
\multicolumn{2}{|c|}{\textbf{Temporal Estimation}} \\
\hline
Learning rate & 0.001 \\

Prediction sequence length & 8 \\

Number of LSTM layeres & 1 \\

Loss Function & MAE \\

Optimizer & Adam \\

Number of Epochs & 50 \\

Batch Size & 32 \\
\hline

\end{tabular}
\label{table:Phase1} % is used to refer this table in the text 
\end{table}
\end{centering}

Fig.~\ref{fig-5} compares spatial estimation performance in terms of load-estimation error measured by MAPE. For the distance-based approach, MAPE is plotted versus the number of neighboring cells \(N\) for two exponents (\(n=3\) and \(n=5\)). Consistent with prior observations in~\cite{salamat}, increasing \(N\) generally raises error because more distant cells contribute weaker information, whereas larger \(n\) reduces error by emphasizing proximity in the weights. For the MLC approach, MAPE is plotted versus the number of clustering levels \(L\); error decreases with \(L\) as clusters become more homogeneous.

\begin{figure}[t]
\centerline{\includegraphics[width = 8.cm ]{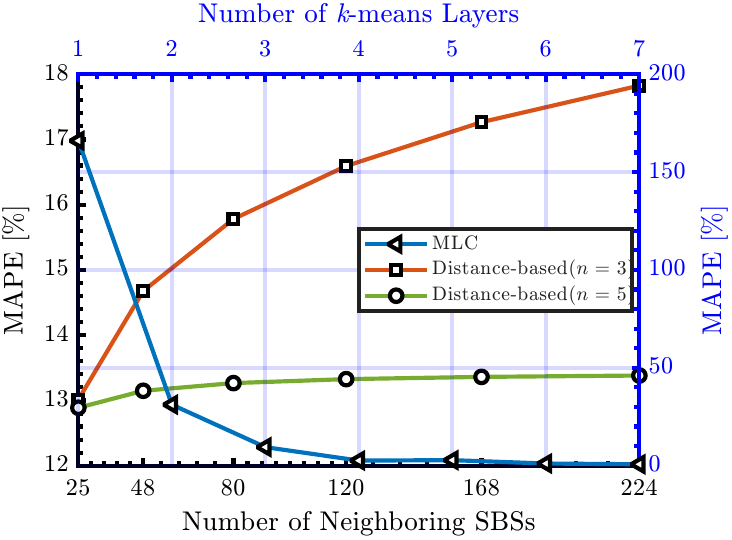}}
\caption{Estimation error comparison of spatial methods. Two $x$- and two $y$-axes are used: the blue axis corresponds to the clustering-based method, and the black axis corresponds to the distance-based method.}
\label{fig-5}
\end{figure}

Fig.~\ref{fig-8} presents MAPE results for the LSTM-based temporal estimation method under various configurations of window size (i.e., number of past time steps) and number of LSTM units.
As shown in the figure, increasing the window size improves prediction accuracy, particularly when paired with a sufficient number of LSTM units. For instance, with 5 LSTM units, MAPE decreases from 4.17\% at a window size of 4 to 1.22\% at a window size of 12.  
This trend becomes more consistent with 10 and 20 LSTM units, where MAPE reaches as low as 0.68\% and 0.64\%, respectively.
These results confirm that LSTM-based temporal modeling is effective in capturing sequential patterns in SBS traffic and highlight the importance of tuning key hyperparameters (e.g., input window size and LSTM capacity) to balance prediction accuracy and real-time feasibility. In particular, larger windows and higher-capacity LSTMs typically improve accuracy but increase inference workload, motivating lightweight configurations when strict real-time constraints apply.

\begin{figure}[t]
\centerline{\includegraphics[width = 8.cm ]{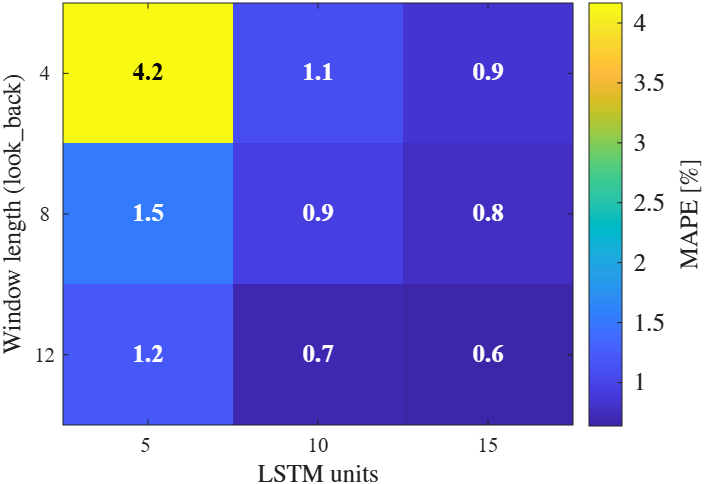}}
\caption{MAPE of LSTM-based traffic load estimation for different window sizes and LSTM units.}
\label{fig-8}
%\vspace{-.5cm} 
\end{figure}

Fig.~\ref{fig-6} plots total network power versus the number of SBSs, $s$, under perfect load knowledge and under estimated loads (by three methods configured as indicated in the legend). For all $s$, total power increases with network size. Across the range, the LSTM-based curve tracks the perfect-knowledge baseline most closely, the MLC curve lies slightly above it, and the distance-based curve shows the largest gap, especially at higher $s$. Each method is run with a single, representative hyperparameter setting selected from prior tuning to yield strong performance; the figure thus compares methods at reasonable operating points rather than presenting hyperparameter sweeps.

%%%%%%%%%%%%%%%%%%%%%%%%%

\begin{figure}[t]
\centerline{\includegraphics[width = 8.cm ]{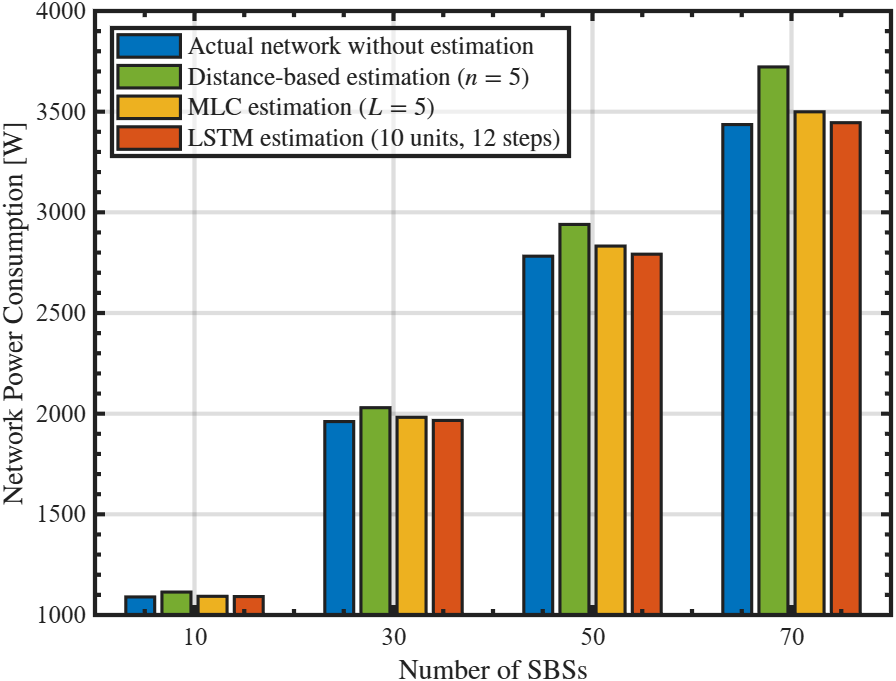}}
\caption{Total network power vs. number of SBSs under perfect load knowledge and under estimated loads.}
\label{fig-6}
%\vspace{-.5cm} 
\end{figure}

Fig.~\ref{fig-7} reports the \emph{decision-change rate}—the percentage of ON/OFF decisions that differ from the perfect-knowledge solution—versus the number of SBSs. Each estimator is run with a single, tuned configuration (as in Fig.~\ref{fig-6}). The rate increases with network size for all methods. The distance-based approach shows the largest divergence (e.g., from 5\% at $s=10$ to 52\% at $s=70$), MLC is intermediate (0–31\%), and LSTM yields the lowest changes (0–18\%). These trends reflect how even modest load-estimation errors can trigger more state flips as the combinatorial decision space grows, consistent with our analysis of over/underestimation effects.

\begin{figure}[t]
\centerline{\includegraphics[width = 8.cm ]{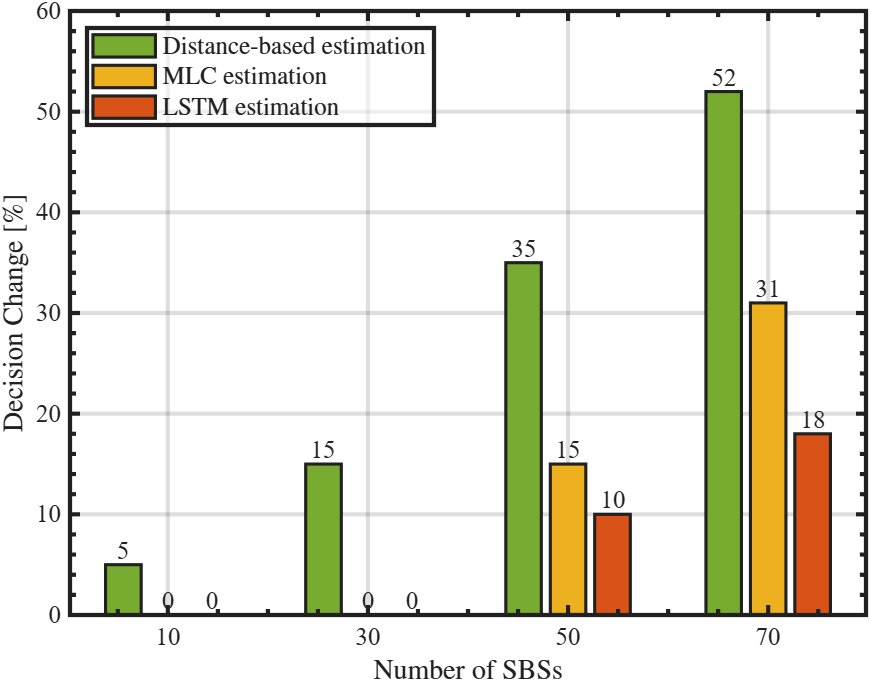}}
\caption{Decision-change rate between ground truth and estimated switching decisions vs. number of SBSs.}
\label{fig-7}
%\vspace{-.5cm} 
\end{figure}

It is important to note that while the results are based on the Milan dataset, we believe these findings offer generic insights applicable across various network scenarios, despite potential dataset-specific limitations. This underscores the broad relevance of our methodologies.
Moreover, Phase~I estimation is performed on a per-SBS basis and is inherently local; thus, the per-target runtime of the distance-based, MLC, and LSTM estimators does not directly scale with the total number of SBSs in the network (network growth mainly increases the number of SBS instances to be supported, not the per-SBS estimation cost).

\subsection{Phase II: Evaluation of Renewable Energy-Aware Cell Switching}

This part of the evaluation focuses on the performance of the proposed renewable energy-aware cell switching framework. The simulation analyzes how solar energy integration affects the overall network power consumption. Solar-capable SBSs are modeled with photovoltaic energy harvesting and battery storage, and the switching optimization accounts for both grid and renewable energy sources.
We examine three operational scenarios, representing a different strategy for incorporating solar-capable SBSs into the optimization process:

\subsubsection{\textbf{Scenario 1:} All SBSs Included in ES}
In this baseline scenario, all SBSs, including solar-capable SBSs, are included in the ES. The algorithm evaluates all possible ON/OFF states, considering only the grid power contributions of solar-capable SBSs. The stored solar energy in the batteries of solar-capable SBSs is not considered or utilized when calculating their power consumption.

\subsubsection{\textbf{Scenario 2:} Solar-Capable SBSs Excluded from ES}
In this scenario, solar-capable SBSs are excluded from ES. The optimization process considers only the grid power contributions of non-solar SBSs. After the optimization algorithm determines the optimal state for non-solar SBSs, the grid power contributions of solar-capable SBSs are added to the total network power consumption.
This approach reduces the complexity of the optimization process by reducing the search space while utilizing the stored solar energy in solar-capable SBSs.

%%%%%%%%%%%%%%%%%%%%%%
\subsubsection{\textbf{Scenario 3:} Threshold-Based Inclusion of Solar-Capable SBSs}
For each solar-capable SBS $j\in S_{\mathrm{r}}$, a state-of-charge ratio
$\rho_j(t)=S_j(t)/S_{\max,j}$ is defined. A tunable threshold $\gamma\in[0,1]$ governs whether $j$
is included in the search space:
\[
\text{include } j \iff \rho_j(t)\le \gamma
\quad\text{(otherwise exclude } j\text{).}
\]
Excluded solar SBSs are treated as fixed \emph{ON} and powered according to the
Renewable Energy Model in Section \ref{subsec:renewable_model} (i.e., first by
$E_{\mathrm{R},j}(t)$, with any deficit drawn from the grid via
$E_{\mathrm{G},j}(t)$). Solar SBSs participate in the ON/OFF optimization
together with non-solar SBSs. This threshold policy prioritizes using harvested
energy while limiting the search space. 

%%%%%%%%%%%%%%%%%
By definition, Scenarios~1 and 2 are special cases of Scenario~3: Scenario~1 corresponds to \(\gamma=1\) (ES over all SBSs), and Scenario~2 corresponds to \(\gamma=0\) (all solar SBSs forced ON, ES only over non-solar SBSs). 
From a computational standpoint, the dominant cost is governed by the size of the ON/OFF decision set. Scenario~2 excludes solar-capable SBSs from the search, and Scenario~3 may exclude a subset of them depending on the state-of-charge threshold, thereby reducing $n$ and significantly shrinking the search space; this search-space reduction principle can also be combined with other optimization engines beyond ES. Renewable-awareness does not change the combinatorial structure of the ON/OFF decision problem; it only adds lightweight per-configuration accounting for the renewable-versus-grid split and storage updates.

\begin{centering}
\begin{table}[t] 
\caption{Simulation Settings for Phase II Renewable Energy-Aware Switching}
\centering % used for centering table 
\begin{tabular}{|c|c|} % centered columns 
  \hline % inserts single horizontal line
Parameters&Value
\\ 
\hline
$P_{\mathrm{MAX} ,j}$&10 kW
\\
$C_s$&0.5 kWh
\\
$\zeta$&95\%
\\
Peak solar hours&08:00 to 18:00
\\
Number of time slots&144
\\
Number of days& 1
\\
\hline
    \end{tabular} \label{table:Phase2} % is used to refer this table in the text 
    %\vspace{-1.5em}
    \end{table}
    \end{centering}

 The total network power consumption as a function of the solar percentage for different numbers of SBSs (i.e., 10, 12, and 14) is illustrated in Fig.~\ref{PvsSolar}. The results compare Scenario 1 (with solar energy) to the network without solar integration. Both curves use the same ES-based switching engine; the baseline disables harvesting/storage so that all SBS energy demand is supplied by the grid, thereby isolating the net gain due solely to renewable awareness.
The network power consumption decreases with increasing solar percentages, demonstrating the effectiveness of renewable energy in reducing grid dependency. Larger networks, such as those with 14 SBSs, exhibit higher total power consumption overall due to their greater aggregate traffic demand and hardware operation. However, they also benefit more significantly from increased solar integration because the absolute amount of harvested renewable energy scales with the number of solar-capable SBSs. As more SBSs are equipped with solar panels, the relative contribution of renewables offsets a larger portion of the grid demand, amplifying the gains in larger configurations. This highlights both the scalability and the compounding benefits of renewable energy-aware solutions in dense deployments.
Networks without solar energy maintain constant power consumption regardless of the solar percentage, providing a reference for comparison. In contrast, networks with solar energy achieve consistently lower power consumption, underscoring the potential of solar energy to enhance network sustainability and efficiency.

%%%%%%%%%%%%%%%%%%%%%%%%
\begin{figure}[t]
\centerline{\includegraphics[width = 8.cm]{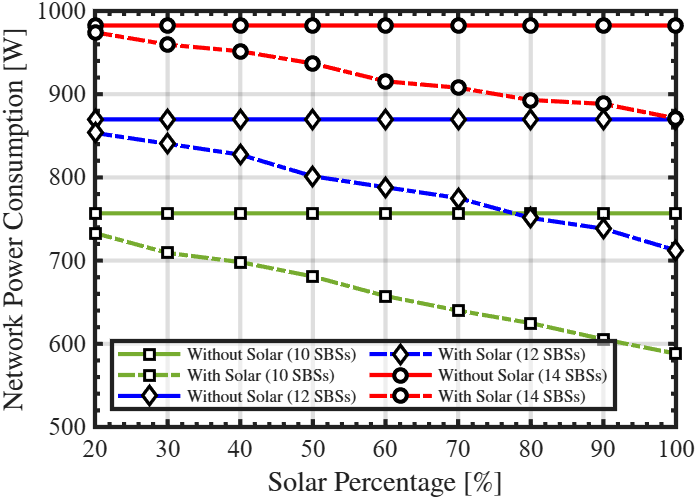}}
\caption{Network power consumption vs. solar percentage for 10, 12, and 14 SBSs, comparing Scenario 1 and the network without solar integration.}
\label{PvsSolar}
%\vspace{-.5cm} 
\end{figure}

Fig.~\ref{Pvstime} presents the hourly power consumption of the network over a 24-hour period for two configurations, with 8 and 14 SBSs, considering Scenario 1. The solar percentages analyzed are 
30\% and 60\%, compared against the network without solar energy integration.
During the peak solar hours, the networks with solar energy exhibit significantly lower power consumption compared to networks without solar energy. This reduction is more pronounced with a higher solar percentage (60\%) due to the increased availability of renewable energy. Outside of peak solar hours, when harvesting is low, and stored energy may be depleted, the solar-enabled curves converge to the no-solar baseline, meaning that the achievable NES is reduced rather than causing infeasibility.
The number of SBSs further impacts power consumption trends. Networks with 14 SBSs show consistently higher power consumption than those with 8 SBSs, as expected due to the greater traffic load. However, the relative reduction in power consumption during solar hours is more noticeable in larger networks because the absolute amount of harvested renewable energy also scales with the number of solar-capable SBSs. As more SBSs contribute to solar energy, a greater share of the grid demand is offset, making the renewable benefits appear more significant in larger configurations. This scalability effect is consistent with the trends observed earlier in Fig.~\ref{PvsSolar}, where larger networks benefited more strongly from increasing solar penetration.

\begin{figure}[t]
\centerline{\includegraphics[width = 8.cm]{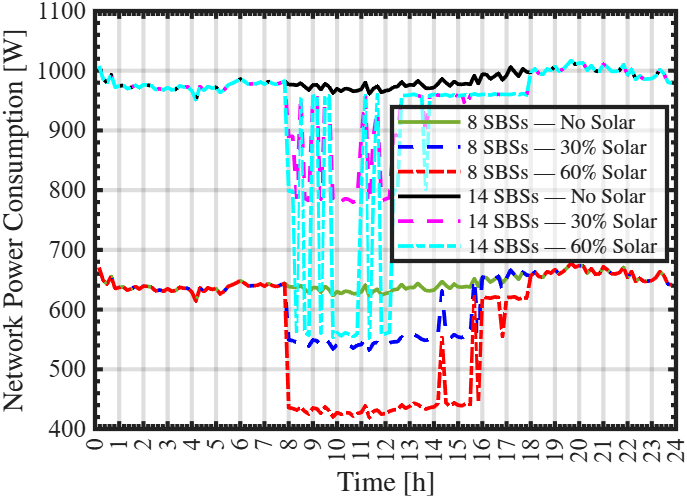}}
\caption{Hourly network power consumption over 24 hours, comparing different solar percentages and SBS configurations under Scenario 1 with a network without solar.}
\label{Pvstime}
%\vspace{-.5cm} 
\end{figure}
The average grid power versus the fraction of solar-capable SBSs is demonstrated in Fig.~\ref{Scenario} for the four solution scenarios. As solar penetration increases, all policies increase NES as solar penetration rises; the separation among curves stems from the different search spaces. Scenario~1 (\(\gamma\!=\!1\)) searches over all SBSs and thus achieves the highest NES (lower grid power) at every penetration level (the lower bound for this family). Scenario~2 (\(\gamma\!=\!0\)) forces all solar SBSs ON and runs ES only over the non-solar subset; it offers the smallest search space but the highest grid power (upper bound). The threshold policy (Scenario~3) with \(\gamma\in\{0.7,0.3\}\) interpolates between these extremes: a larger threshold (\(\gamma=0.7\)) includes more solar SBSs in ES, closely tracking Scenario~1 while already reducing the number of ES variables; a smaller threshold (\(\gamma=0.3\)) further reduces the search dimension at a visible cost in grid power. Let
$\mathcal{V}$ be the set of SBSs that remain decision variables in the ES (i.e., those not forced ON). Since ES complexity scales as 
\(2^{|\mathcal{V}|}\), lowering \(\gamma\) trades optimality for exponential savings in computation.

\begin{figure}[t]
\centerline{\includegraphics[width = 8.cm]{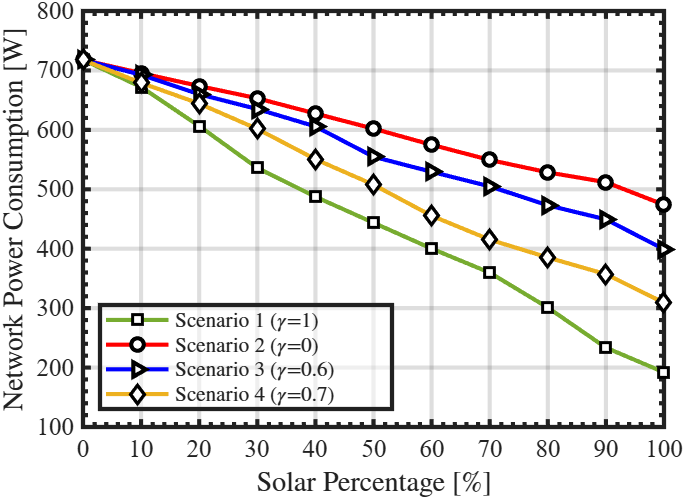}}
\caption{Total network power consumption as a function of solar percentage under three different cell switching scenarios.}

\label{Scenario}
%\vspace{-.5cm} 
\end{figure}

\section{Conclusion}\label{sec:conclusion}

This study addresses two critical challenges for optimizing power consumption in vHetNets: estimating traffic loads for SBSs in sleep mode and integrating renewable energy into cell switching strategies.
For traffic load estimation in Phase I, we developed mathematical frameworks to explain the behaviors of spatial interpolation methods and designed a portfolio of estimators with different levels of data dependency. This flexible design enables the framework to adapt to various deployment conditions, from data-scarce environments to data-rich networks, ensuring practical applicability. Validated with the real Milan dataset, these methods have proven effective in accurately estimating traffic loads, enabling the implementation of more efficient energy-saving strategies.
In Phase II, we proposed renewable energy–aware cell switching strategies to further enhance network sustainability. we proposed renewable-aware cell switching with three operational scenarios that capture hybrid grid-renewable policies. This scenario-based flexibility allows the framework to adapt its optimization strategy to different levels of solar penetration and storage availability. Simulation results showed significant reductions in power consumption, particularly in scenarios leveraging stored solar energy and selective SBS inclusion in the optimization process.
Overall, combining accurate sleeping-cell load estimation with renewable-aware switching yields substantial NES, providing a practical pathway to sustainable vHetNets.

\ifCLASSOPTIONcaptionsoff
  \newpage
\fi
\bibliographystyle{IEEEtran}
\small
\bibliography{IEEEabrv,Bibliography}

\begin{thebibliography}{10}
\providecommand{\url}[1]{#1}
\csname url@samestyle\endcsname
\providecommand{\newblock}{\relax}
\providecommand{\bibinfo}[2]{#2}
\providecommand{\BIBentrySTDinterwordspacing}{\spaceskip=0pt\relax}
\providecommand{\BIBentryALTinterwordstretchfactor}{4}
\providecommand{\BIBentryALTinterwordspacing}{\spaceskip=\fontdimen2\font plus
\BIBentryALTinterwordstretchfactor\fontdimen3\font minus \fontdimen4\font\relax}
\providecommand{\BIBforeignlanguage}[2]{{%
\expandafter\ifx\csname l@#1\endcsname\relax
\typeout{** WARNING: IEEEtran.bst: No hyphenation pattern has been}%
\typeout{** loaded for the language `#1'. Using the pattern for}%
\typeout{** the default language instead.}%
\else
\language=\csname l@#1\endcsname
\fi
#2}}
\providecommand{\BIBdecl}{\relax}
\BIBdecl

\bibitem{bs_power}
M.~Feng, S.~Mao, and T.~Jiang, ``Base station on-off switching in 5{G} wireless networks: Approaches and challenges,'' \emph{IEEE Wireless Communications}, vol.~24, no.~4, pp. 46--54, Aug. 2017.

\bibitem{2222222}
C.-L. I., S.~Han, and S.~Bian, ``Energy-efficient {5G} for a greener future,'' \emph{Nature Electronics}, vol.~3, pp. 182--184, Apr. 2020.

\bibitem{itu_vision_june_23}
``{The ITU-R Framework for IMT-2030},'' ITU-R Working Party 5D, International Telecommunication Union (ITU), Tech. Rep., Jul. 2023.

\bibitem{Belkhir2018}
L.~Belkhir and A.~Elmeligi, ``Assessing {ICT} global emissions footprint: Trends to 2040 \& recommendations,'' \emph{Journal of Cleaner Production}, vol. 177, pp. 448--463, 2018.

\bibitem{3GPP}
{{3rd Generation Partnership Project (3GPP)}}, ``{Study on network energy savings for NR},'' {3GPP}, {3GPP Technical Report} {TR 38.864}, 2023.

\bibitem{6512843}
J.~B. Rao and A.~O. Fapojuwo, ``A survey of energy efficient resource management techniques for multicell cellular networks,'' \emph{IEEE Communications Surveys and Tutorials}, vol.~16, no.~1, pp. 154--180, May 2014.

\bibitem{ELAA2022JR}
A.~E. Amine, J.-P. Chaiban, H.~A.~H. Hassan, P.~Dini, L.~Nuaymi, and R.~Achkar, ``Energy optimization with multi-sleeping control in {5G} heterogeneous networks using reinforcement learning,'' \emph{IEEE Transactions on Network and Service Management}, vol.~19, no.~4, pp. 4310--4322, Mar. 2022.

\bibitem{EOMK2017JR}
M.~W. Kang and Y.~W. Chung, ``An efficient energy saving scheme for base stations in {5G} networks with separated data and control planes using particle swarm optimization,'' \emph{Energies}, vol.~10, no.~9, Sep. 2017.

\bibitem{Metin_VFA_CellSwitch}
M.~Ozturk, A.~I. Abubakar, J.~P.~B. Nadas, R.~N.~B. Rais, S.~Hussain, and M.~A. Imran, ``Energy optimization in ultra-dense radio access networks via traffic-aware cell switching,'' \emph{IEEE Transactions on Green Communications and Networking}, vol.~5, no.~2, pp. 832--845, Feb. 2021.

\bibitem{8735834}
J.~Ye and Y.-J.~A. Zhang, ``{DRAG}: Deep reinforcement learning based base station activation in heterogeneous networks,'' \emph{IEEE Transactions on Mobile Computing}, vol.~19, no.~9, pp. 2076--2087, Dec. 2020.

\bibitem{ACM}
M.~Marsan and M.~Meo, ``Energy efficient management of two cellular access networks,'' \emph{SIGMETRICS Performance Evaluation Review}, vol.~37, pp. 69--73, Mar. 2010.

\bibitem{ALOTAIBI2025104213}
J.~Alotaibi, O.~S. Oubbati, M.~Atiquzzaman, F.~Alromithy, and M.~R. Altimania, ``Optimizing disaster response with {UAV}-mounted {RIS} and {HAP}-enabled edge computing in {6G} networks,'' \emph{Journal of Network and Computer Applications}, vol. 241, Sep. 2025.

\bibitem{8024181}
M.~Feng, S.~Mao, and T.~Jiang, ``{BOOST}: Base station on-off switching strategy for green massive {MIMO} {HetNets},'' \emph{IEEE Transactions on Wireless Communications}, vol.~16, no.~11, pp. 7319--7332, Nov. 2017.

\bibitem{10061602}
X.~Tan, K.~Xiong, B.~Gao, P.~Fan, and K.~B. Letaief, ``Energy-efficient base station switching-off with guaranteed cooperative profit gain of mobile network operators,'' \emph{IEEE Transactions on Green Communications and Networking}, vol.~7, no.~3, pp. 1250--1266, Sep. 2023.

\bibitem{9528008}
J.~Lin, Y.~Chen, H.~Zheng, M.~Ding, P.~Cheng, and L.~Hanzo, ``A data-driven base station sleeping strategy based on traffic prediction,'' \emph{IEEE Transactions on Network Science and Engineering}, vol.~11, no.~6, pp. 5627--5643, Nov.-Dec. 2024.

\bibitem{10506912}
H.~Nashaat, N.~H. Mohammed, S.~M. Abdel-Mageid, and R.~Y. Rizk, ``Machine learning-based cellular traffic prediction using data reduction techniques,'' \emph{IEEE Access}, vol.~12, pp. 58\,927--58\,939, 2024.

\bibitem{eswa}
\BIBentryALTinterwordspacing
W.~Jiang, ``Cellular traffic prediction with machine learning: A survey,'' \emph{Expert Syst. Appl.}, vol. 201, no.~C, Sep. 2022. [Online]. Available: \url{https://doi.org/10.1016/j.eswa.2022.117163}
\BIBentrySTDinterwordspacing

\bibitem{jsan9040053}
\BIBentryALTinterwordspacing
B.~Mahdy, H.~Abbas, H.~S. Hassanein, A.~Noureldin, and H.~Abou-zeid, ``A clustering-driven approach to predict the traffic load of mobile networks for the analysis of base stations deployment,'' \emph{Journal of Sensor and Actuator Networks}, vol.~9, no.~4, Nov. 2020. [Online]. Available: \url{https://www.mdpi.com/2224-2708/9/4/53}
\BIBentrySTDinterwordspacing

\bibitem{8502042}
M.~D’Amours, A.~Girard, and B.~Sansò, ``Planning solar in energy-managed cellular networks,'' \emph{IEEE Access}, vol.~6, pp. 65\,212--65\,226, Oct. 2018.

\bibitem{6874568}
Y.-K. Chia, S.~Sun, and R.~Zhang, ``Energy cooperation in cellular networks with renewable powered base stations,'' \emph{IEEE Transactions on Wireless Communications}, vol.~13, no.~12, pp. 6996--7010, Aug. 2014.

\bibitem{9773096}
D.~Renga and M.~Meo, ``Can high altitude platform stations make {6G} sustainable?'' \emph{IEEE Communications Magazine}, vol.~60, no.~9, pp. 75--80, May 2022.

\bibitem{10938203}
\BIBentryALTinterwordspacing
M.~Ozturk, M.~Salamatmoghadasi, and H.~Yanikomeroglu, ``Integrating terrestrial and non-terrestrial networks for sustainable 6{G} operations: A latency-aware multi-tier cell-switching approach,'' \emph{IEEE Network (Early Access)}, pp. 1--1, 2025. [Online]. Available: \url{https://arxiv.org/abs/2508.10849}
\BIBentrySTDinterwordspacing

\bibitem{9380673}
G.~Karabulut~Kurt, M.~G. Khoshkholgh, S.~Alfattani, A.~Ibrahim, T.~S.~J. Darwish, M.~S. Alam, H.~Yanikomeroglu, and A.~Yongacoglu, ``A vision and framework for the high altitude platform station {(HAPS)} networks of the future,'' \emph{IEEE Communications Surveys \& Tutorials}, vol.~23, no.~2, pp. 729--779, 1st Quart., 2021.

\bibitem{itu_2000}
``{Preferred characteristics of systems in the fixed service using high altitude platforms operating in the bands 47.2-47.5 GHz and 47.9- 48.2 GHz,},'' Int. Telecommun. Union, Switzerland, ITU-Recommendation F.1500, Tech. Rep., Jan. 2000.

\bibitem{6056691}
G.~Auer, V.~Giannini, C.~Desset, I.~Godor, P.~Skillermark, M.~Olsson, M.~A. Imran, D.~Sabella, M.~J. Gonzalez, O.~Blume, and A.~Fehske, ``How much energy is needed to run a wireless network?'' \emph{IEEE Wireless Communications}, vol.~18, no.~5, pp. 40--49, Oct. 2011.

\bibitem{7925662}
H.~Wu, X.~Xu, Y.~Sun, and A.~Li, ``Energy efficient base station on/off with user association under {C/U} split,'' \emph{2017 IEEE Wireless Communications and Networking Conference}, pp. 1--6, May 2017.

\bibitem{maryam}
M.~Salamatmoghadasi, A.~Mehrabian, and H.~Yanikomeroglu, ``Energy sustainability in dense radio access networks via high altitude platform stations,'' \emph{IEEE Networking Letters}, vol.~6, no.~1, pp. 21--25, Mar. 2024.

\bibitem{HETS2023P}
T.~Song, D.~Lopez, M.~Meo, N.~Piovesan, and D.~Renga, ``High altitude platform stations: the new network energy efficiency enabler in the 6{G} era,'' \emph{2024 IEEE Wireless Communications and Networking Conference (WCNC)}, pp. 1--6, Jul. 2024.

\bibitem{11352980}
M.~Salamatmoghadasi, A.~Mehrabian, H.~Yanikomeroglu, and G.~Kaddoum, ``Sustainable vertical heterogeneous networks: A cell switching approach with high altitude platform station,'' \emph{IEEE Transactions on Green Communications and Networking}, vol.~10, pp. 1951--1966, Jan. 2026.

\bibitem{article}
H.~Zhao, ``Research on improvement and parallelization of k-means clustering algorithm,'' \emph{IEEE 3rd International Conference on Frontiers Technology of Information and Computer}, pp. 57--61, Dec. 2021.

\bibitem{8585453}
O.~S. Oubbati, N.~Chaib, A.~Lakas, and S.~Bitam, ``On-demand routing for urban {VANET}s using cooperating {UAV}s,'' \emph{2018 International Conference on Smart Communications in Network Technologies (SaCoNeT)}, pp. 108--113, Oct. 2018.

\bibitem{10976414}
O.~S. Oubbati, J.~Alotaibi, F.~Alromithy, M.~Atiquzzaman, and M.~R. Altimania, ``A {UAV}-{UGV} cooperative system: Patrolling and energy management for urban monitoring,'' \emph{IEEE Transactions on Vehicular Technology}, vol.~74, no.~9, pp. 13\,521--13\,536, Apr. 2025.

\bibitem{DVN/EGZHFV_2015}
\BIBentryALTinterwordspacing
T.~Italia, ``{Telecommunications - SMS, Call, Internet - MI},'' 2015. [Online]. Available: \url{https://doi.org/10.7910/DVN/EGZHFV}
\BIBentrySTDinterwordspacing

\bibitem{salamat}
M.~Salamatmoghadasi, M.~Ozturk, and H.~Yanikomeroglu, ``Enhancing sustainability in {HAPS}-assisted {6G} networks: Load estimation aware cell switching,'' \emph{IEEE International Symposium on Personal, Indoor and Mobile Radio Communications ({PIMRC})}, pp. 1--6, Sep. 2025.

\end{thebibliography}

\vfill
\end{document}